\documentclass[english,11pt, letter]{article}

\usepackage{amsmath,amsthm,amssymb,mathtools,geometry}
\usepackage[active]{srcltx}
\usepackage{xcolor,helvet,courier,esint,tabularx,comment,hyperref,float,inputenc,calc,float}
\hypersetup{colorlinks=true,citecolor=red}
\usepackage{algorithm,algpseudocode}
\usepackage{graphicx}
\usepackage{subcaption}
\usepackage{soul}
\hypersetup{colorlinks=false}
\usepackage[colorinlistoftodos]{todonotes}

\DeclareMathOperator*{\argmax}{arg\,max}
\DeclareMathOperator*{\argmin}{arg\,min}

\pdfoutput=1

\newtheorem{theorem}{Theorem}
\newtheorem{lemma}{Lemma}
\newtheorem{definition}{Definition}

\author{
Karine Chubarian\\
University of Illinois Chicago\\
\texttt{kchuba2@uic.edu} \\
\and
Anastasios Sidiropoulos\\
University of Illinois Chicago\\
\texttt{sidiropo@uic.edu}
}

\newcommand{\R}{\mathbb{R}}
\newcommand{\reals}{\mathbb{R}}

\newcommand{\length}{\mathsf{length}}

\newcommand{\etal}{\emph{et al. }}
\newcommand{\poly}{\mathsf{poly}}
\newcommand{\safe}{\mathsf{safe}}
\newcommand{\exposed}{\mathsf{exposed}}
\newcommand{\inner}{\mathsf{inner}}
\newcommand{\outerr}{\mathsf{outer}}
\newcommand{\density}{\mathsf{density}}
\newcommand{\forest}{\mathsf{forest}}
\newcommand{\tripod}{\mathsf{tripod}}
\newcommand{\metspace}{\mathcal{M}}
\newcommand{\ball}{\mathsf{Ball}}
\newcommand{\distortion}{\mathsf{distortion}}

\let\oldReturn\Return
\renewcommand{\Return}{\State\oldReturn}

\date{}
\title{Computing Bi-Lipschitz Outlier Embeddings into the Line}

\begin{document}
\maketitle

\begin{abstract}
The problem of computing a bi-Lipschitz embedding of a graphical metric into the line with minimum distortion has received a lot of attention.
The best-known approximation algorithm computes an embedding with distortion $O(c^2)$, where $c$ denotes the optimal distortion [B\u{a}doiu \etal~2005].
We present a bi-criteria approximation algorithm that extends the above results to the setting of \emph{outliers}.

Specifically, we say that a metric space $(X,\rho)$ admits a $(k,c)$-embedding if there exists $K\subset X$, with $|K|=k$, such that $(X\setminus K, \rho)$ admits an embedding into the line with distortion at most $c$.
Given $k\geq 0$, and a metric space that admits a $(k,c)$-embedding, for some $c\geq 1$, our algorithm computes a $(\poly(k, c, \log n), \poly(c))$-embedding in polynomial time.
This is the first algorithmic result for outlier bi-Lipschitz embeddings.
Prior to our work, comparable outlier embeddings where known only for the case of additive distortion.
\end{abstract}

\section{Introduction}

The theory of metric embeddings provides an extensive toolbox that has found applications in several geometric data-analytic tasks.
At the high level, an embedding of a metric space ${\cal M}=(X,\rho)$ into some metric space ${\cal M}'=(X',\rho)$ is a mapping $f:X\to X'$ that preserves certain interesting geometric properties of ${\cal M}$.
In most cases, it is desirable to obtain embeddings that minimize some notion of \emph{distortion}.

Despite the success of metric embeddings methods in several application domains, one significant limitation of most existing methods is that they are not robust to noise in the form of outlier points in the input.
This setting is of particular interest in the case where the data does not perfectly fit the underlying geometric model, or when some points are corrupted due to measurement errors.
The outlier model also has connections to the setting of adversarial machine learning 
\cite{goodfellow2018making}.
More specifically, in the setting of \emph{poisoning attacks}, it is often assumed that a small subset of the training data set is corrupted adversarially.
For example, in a classification application, some of the training samples can be modified arbitrarily.
Therefore, it is important to design data-analytic primitives that are robust against this type of adversarial input perturbation.

Our aim is to bypass the limitations of current metric embedding methods by designing approximation algorithms that given some input space ${\cal M}$, they compute a small subset of points to delete, and an embedding of the residual space into some desired host space.

\subsection{Our contribution}

%Now we present the main result of our paper. Let $G$ be an undirected graph $G=(V(G),E(G))$ with unit weights assigned to all $e\in E(G)$. Given that $G$ admits a $(k,c)$-embedding into the line for some $k>0,c>1$, we consider the problem of computing in polynomial time a $(\Tilde{O}(\poly(c)\cdot k),\poly(c))$-embedding of $G$ into the line. We obtain a following result which is proven in \ref{sect:main_algo}.

%\subsection{Metric embeddings}
%\textbf{Metric embeddings.}
We now formally define outlier embeddings and state our main result.
Let $\metspace=(X,\rho)$, $\metspace'=(X',\rho')$ be metric spaces. An injection $f: X\to X'$ is called an \emph{embedding}.
Given an embedding $f$, its  \emph{distortion} is defined as
\begin{align*}
\distortion(f)&= \sup_{x\neq y\in X}\frac{\rho'\left(f(x),f(y)\right)}{\rho(x,y)} \cdot \sup_{x'\neq y'\in X}\frac{\rho(x',y')}{\rho'\left(f(x'),f(y')\right)}.
\end{align*}
We also refer to this notion of distortion as \emph{multiplicative distortion}.
An embedding is \emph{bi-Lipschitz} if its distortion is bounded.
When $\metspace'=(\reals,\ell_2)$ then we say that $\metspace$ admits an embedding into the line.
If $\distortion(f)\leq c$, then we say that $f$ is a \emph{$c$-embedding}.
We use the following definition for outlier embeddings (see also \cite{DBLP:conf/soda/SidiropoulosWW17}).

%\begin{definition}[$c$-embedding]
%A metric space $\metspace=(X,\rho)$ \emph{admits a $c$-embedding} into another metric space $\metspace'=(X',\rho')$ for some $c\ge 1$ if there exists an embedding $f:X\to X'$ with $\distortion(f)\le c$.
%\end{definition}
%In particular, when $\metspace'=(\reals,\ell_2)$ then we say that $\metspace$ admits a $c$-embedding into the line.

%\begin{definition}[$(k,c)$-embedding]
A metric space $\metspace=(X,\rho)$ \emph{admits a $(k,c)$-embedding} into another metric space $\metspace'=(X',\rho')$ for some $c\ge 1$, $k\ge 0$ if there exists $K\subseteq X$, with $|K|\le k$, and $f:X\setminus K\to X'$, with $\distortion(f)\le c$.
We say that such $K\subseteq X$ is an \emph{outlier set} (w.r.t.~$f$).
%\end{definition}

In the present work, we focus on the case where the input metric space is the shortest-path metric of an unweighted graph, and the host space is the real line.
This setting, but without outliers, has been studied extensively in the literature (see Section \ref{sec:related} for a more detailed discussion).
The shortest-path metrics of unweighted graphs arise naturally in applications, for example, when considering the $k$-NN graph of a point set;
that is, by taking the set of vertices to be a set of samples from some unknown manifold, and the edge set to be all pairs $\{u,v\}$, where $u$ is one of the $k$ nearest neighbors of $v$.
Moreover, the case of embedding into the real line is a prototypical mathematical model for the problem of discovering 1-dimensional structure in a metrical data set.

The following summarizes the main result of this paper.

\begin{theorem}\label{thm:main_thm}
Let $G$ be a graph, $k\geq 0$, $c\geq 1$.
There exists a polynomial-time algorithm which given $G$, $k$, and $c$, terminates with exactly one of the following outcomes:
\begin{description}
\item{(1)}
Correctly decides that $G$ does not admit a $(k,c)$-embedding into the line.
\item{(2)}
Computes a $(O(c^6k\log^{5/2}{n}),O(c^{13}))$-embedding of $G$ into the line.
\end{description}
\end{theorem}

\subsection{Related work}\label{sec:related}

Low-distortion metric embeddings have been studied extensively within mathematics and computer science.
We refer the reader to \cite{indyk20178} for a detailed exposition of the work that is of main interest for computer science.
Here, we discuss some results that are most relevant to our work.

\textbf{Approximation algorithms.}
The problem of computing an embedding of some input metric space ${\cal M}$ into some host space ${\cal M}'$ with approximately minimum distortion has received a lot of attention.
Most positive results are concerned with the case where ${\cal M}'$ is the line, or, more generally, some $1$-dimensional space.
Specifically, B\u{a}doiu \etal~\cite{badoiu2005approximation} obtained an algorithm which given an unweighted graph that admits a $c$-embedding into the line, computes a $O(c^3)$-embedding into the line.
Approximation algorithms have also been obtained by B\u{a}doiu \etal~\cite{badoiu2005low} for the case where the input is a weighted tree, 
and by Nayyeri and Raichel \cite{nayyeri2015reality}
for the case where the input is a general metric space.

Approximation algorithms for embedding into more general $1$-dimensional spaces have also been considered.
B\u{a}doiu \etal~\cite{DBLP:conf/soda/BadoiuIS07} consider the case where the host space is a tree, 
Chepoi \etal~\cite{chepoi2012constant} consider the case where the host space is an outerplanar graph,
and 
Nayyeri and Raichel \cite{nayyeri2017treehouse}~generalize this to the case where the host space is a graph of bounded treewidth.
Carpenter \etal~\cite{DBLP:conf/compgeom/CarpenterFL0S18} obtain an approximation algorithm for embedding unweighted graphs into subdivisions of any fixed ``pattern'' graph $H$ (embedding into the line corresponds to the case where $H$ is a single edge, while embedding into a cycle is the case where $H$ is a triangle).

The case of higher-dimensional host spaces appears to be significantly more challenging.
The only positive results are an approximation algorithm for embedding finite subsets of the 2-sphere into $\mathbb{R}^2$ \cite{badoiu2005approximation}, and approximation algorithms for embedding ultrametrics into $\mathbb{R}^d$ \cite{badoiu2006embedding,de2010fat}.
On the negative side, it has been shown that for any $d\geq 1$, the problem of embedding into $d$-dimensional Euclidean space with minimum distortion is hard to approximate within a factor of $n^{\alpha/d}$, for some constant $\alpha>0$ (the case $d=1$ is due to \cite{badoiu2005low} and the case $d\geq 2$ is due to \cite{DBLP:conf/focs/MatousekS08}).

\textbf{FPT algorithms.}
The problem of computing an embedding into the line parameterized by the optimal distortion has also been considered.
Fellows \etal\cite{DBLP:conf/icalp/FellowsFLLRS09} gave an FPT algorithm for embedding unweighted graphs into the line.
A nearly-matching lower bound on the running time (assuming ETH) was obtained by Lokshtanov \etal~\cite{DBLP:conf/soda/LokshtanovMS11}.
FTP algorithms for embedding unweighted graphs into subdivisions of an arbitrary fixed pattern graph $H$ have also been obtained by Carpenter \etal~\cite{DBLP:conf/compgeom/CarpenterFL0S18}.

\textbf{Outlier embeddings.}
The problem of computing outlier embeddings was introduced by Sidiropoulos \etal~\cite{DBLP:conf/soda/SidiropoulosWW17}.
They considered the case of embedding into $d$-dimensional Euclidean space, and into trees.
The main difference with our work is that \cite{DBLP:conf/soda/SidiropoulosWW17} deals with the case of \emph{additive} distortion, while we are concerned with multiplicative distortion.
As a result, the results in \cite{DBLP:conf/soda/SidiropoulosWW17} are incomparable to ours.
We remark, however, that the case of mutliplicative distortion is known to be significantly more challenging.
To the best of our knowledge, our result is the first non-trivial upper bound for computing outlier embeddings minimizing the multiplicative distortion.

%One can relax the question and allow presence of outliers in $\metspace$. The problem of approximating the \emph{additive} distortion in presence of outliers has been studied by Sidiropoulos~\etal \cite{DBLP:conf/soda/SidiropoulosWW17}. For an $n$-point metric space $\metspace$ which admits a $(c,k)$-embedding into a metric space $\metspace'$ they discovered polynomial-time algorithms which compute $(O(c\log n),3k)$-, $(O(c\log n), 4k)$-embeddings for $\metspace'$ being an ultrametric or a tree respectively. Furthermore, they gave $O(n^{d+3})$ algorithm which provides $(O(\sqrt{c}),2k)$-embedding when $\metspace'$ is a $d$-dimensional Eucledian space.

\subsection{High-level overview of the algorithm}
We now give an informal description of our algorithm, highlighting the main technical challenges. The input  consists of an undirected graph $G$ and some $k\ge 0$, $c\ge 1$. The algorithm either correctly decides that there exists no $(k,c)$-embedding of $G$ into the line, or outputs a $(k',c')$-embedding of $G$ into the line, for some $k'=\poly(k,c,\log n)$, $c'=\poly(c)$.
%$(\poly)$-approximation of (some) optimal embedding of $G$.

The crux of the algorithm is to identify and remove three ``obsrtuctions'' for low-distortion embeddability into the line.
These three obstructions are \emph{regions of high density}, \emph{large metrical cycles} and \emph{large metrical tripods}.
We next discuss the steps used to handle each one of these obstructions, and describe how all the steps are combined in the final algorithm.

\textbf{Obstruction 1: Reducing the density.}
The \emph{density} of a graph is defined to be
\[
\Delta(G)=\max_{v\in V(G), R\in \mathbb{N}}\frac{|\ball_G(v,R)|-1}{2R}.
\]
It is known that the density of any graph that admits a $c$-embedding into the line is $O(c)$ \cite{badoiu2005approximation}.
%Intuitively, the vertices of the graph are "packed too tightly" to admit "nice" embedding into the line.
Therefore, if $G$ admits a $(k,c)$-embedding, then there must exist some set of at most $k$ vertices, whose deletion leaves a graph with density $O(c)$.
We observe that the density of a graph is a hereditary property, meaning that for any $H\subseteq G$, we have $\Delta(H)\leq \Delta(G)$.
This leads to a following recursive procedure: if the density is higher than $O(c)$, we compute a balanced vertex separator $X\subseteq V(G)$, and recurse on $G\setminus X$.
We set 
\[
K_{\density}:=\bigcup_{\text{all separators } X} X.
\]
Let us also denote $G\setminus K_{\density}$ as $G'$.
It is immediate that $\Delta(G')=O(c)$, and we show that $|K_{\density}|=\poly(k,c,\log n)$.

\textbf{Obstruction 2: Eliminating large metrical cycles.}
%The next obstruction for embeddability into the line is large ``metrical cycles''.
It is known that any embedding of the $n$-cycle into the line must incur distortion $\Omega(n)$ \cite{badoiu2005approximation}.
More generally, it is possible to define an obstruction, which we refer to as a \emph{metrical cycle}, and which contains cycles as a special case, but allows for more general shortest-path distances (see Figure \ref{fig:metrical_cycle}).
We show how to delete a small number of vertices so that the resulting graph does not contain any large metrical cycles, and then we find a low-distortion embedding into some forest.
%Intuitively, one can see that if $G$ is a simple cycle on $n$ vertices, then any embedding of $G$ into the line has distortion $\Omega(n)$.
%In general, we want to eliminate structures resembling a simple cycle on some ''topological'' level (for a formal definition, see ''Obstruction to embeddability'' in \cite{DBLP:conf/soda/BadoiuIS07}). 

\begin{figure}[h]
\def\svgscale{0.37}
\centering
   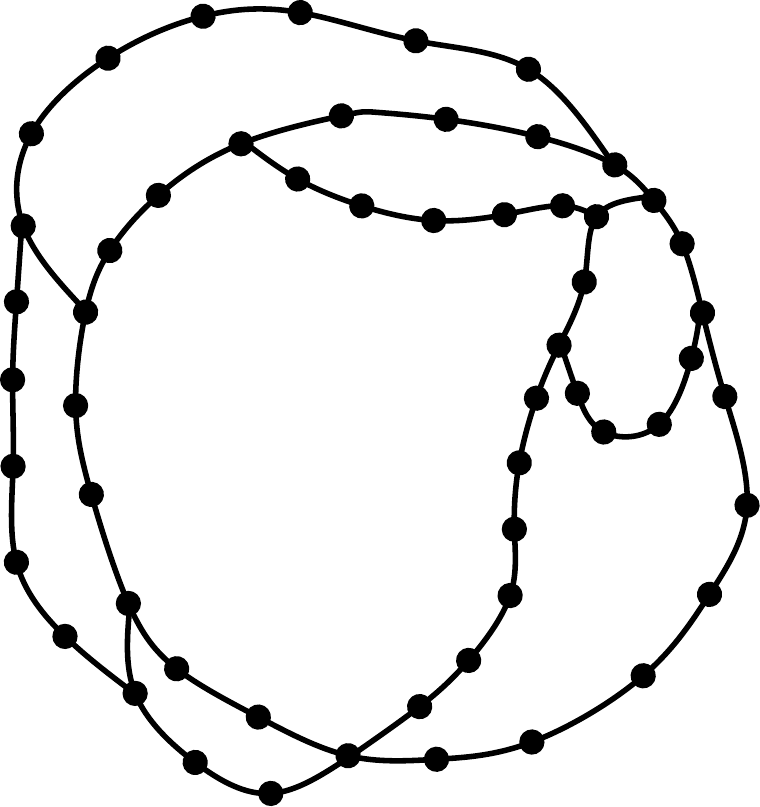
   \caption{Example of a large metrical cycle.}
   \label{fig:metrical_cycle}
\end{figure}

We now briefly describe the procedure for eliminating large metrical cycles.
We start by computing a $\poly(c)$-net $N$ in $G'$.
We then find a Voronoi partition $\cal P$ centered at $N$:
for any vertex $v\in G'$, we assign $v$ to a cluster centered at its nearest neighbour $y\in N$ (we break ties to ensure connectivity).
Let $H$ be the minor of $G$ obtained by contracting each cluster to its center $y\in N$.
We compute an approximate minimum feedback vertex set $Y$ in $H$.
We set 
\[
K_{\forest}:=\bigcup_{x\in Y}{\cal P}(x),
\]
and $G''=G' \setminus K_{\forest}$.
Note that the low density of $G'$ ensures that $|K_{\forest}|$ is small. 
Furthermore, we show that $G''$ admits a low-distortion embedding into a forest.

\begin{figure}[h!]
    \begin{center}
         \includegraphics[height=2.6cm]{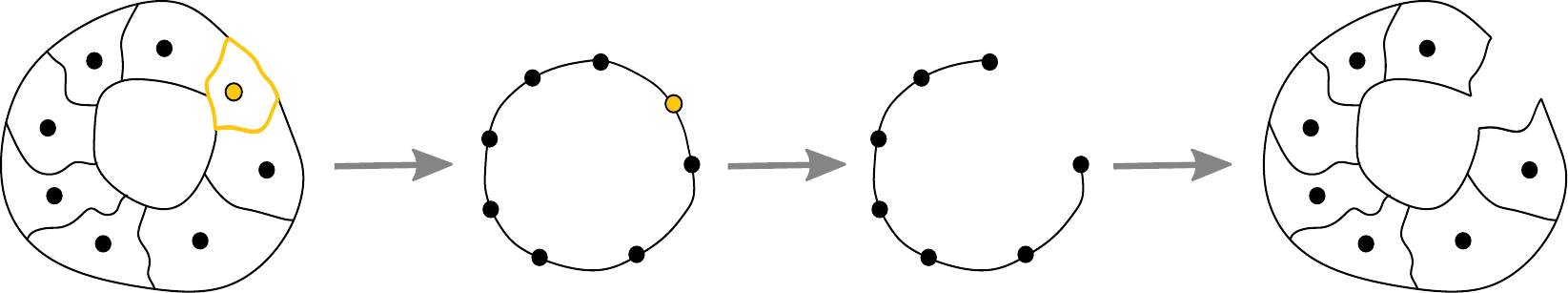}   
    \end{center}
\caption{Elimination of large metrical cycles. From left to right: the graph $G'$, the minor $H$, the forest $H \setminus Y$, and the graph $G''$.}
\end{figure}

\textbf{Obstruction 3: Eliminating large metrical tripods.}
A \emph{tripod} is a tree consisting of the union of three paths with a common endpoint;
we say that a tripod is \emph{$R$}-large if the length of each of the three paths is at least $R$.
Any embedding of a $R$-large tripod into the line must incur distortion $\Omega(R)$.
We show how to delete a small number of vertices so that the resulting graph does not have any subgraphs with a shortest-path metric that resembles that of a $\Omega(\poly(c))$-large tripod. 
More specifically, via a reduction to the Minimum Set Cover problem, we compute some $Z\subseteq V(H\setminus Y)$, so that the forest $H\setminus (Y\cup Z)$ does not contain any $\Omega(\poly(c))$-large tripods (see Figure \ref{fig:tripod_removal}).
We set
\[
K_{\tripod}:=\bigcup_{w\in (H\setminus Y)\setminus Z}{\cal P}(w).
\]
and $G''=G'\setminus K_{\tripod}$.
Since the forest $H\setminus (Y\cup Z)$ does not contain any large tripods, we can show that it admits a low-distortion embedding into the line.
Furthermore, we can use this embedding to also embed $G''$ into the line.

%We show that if $G''$ admits a low-distortion embedding into the line then $H\setminus Y$ contains no $O(\poly(c))$-large tripods. Thus, we can enumerate all $O(n^4)$ tripods of size $O(\poly(c))$ and create an instance of a set cover where the covering subsets of the universe are the vertices of the tripods.

\begin{figure}[H]
\def\svgscale{0.7}
\centering
   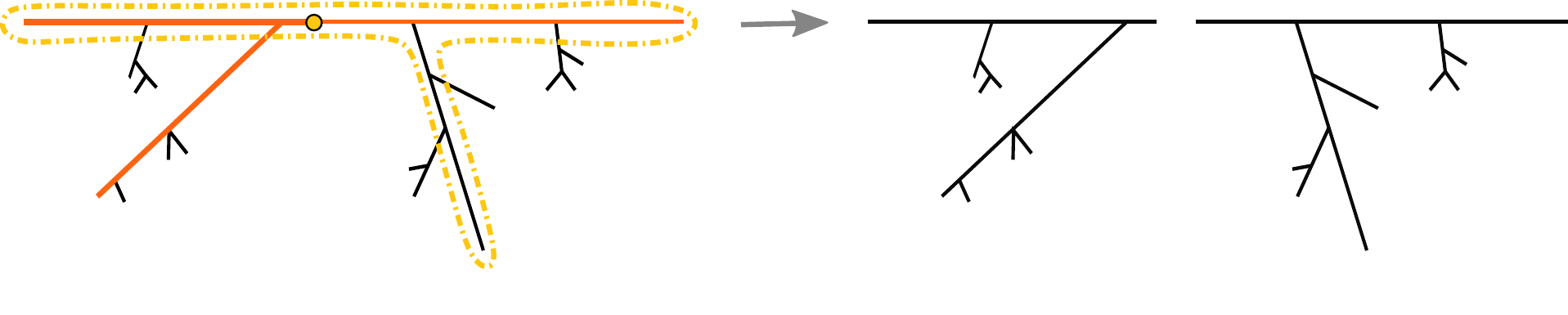
   \caption{Elimination of a large tripod. A yellow vertex removes the red tripod and the yellow dotted tripod simultaneously.}
   \label{fig:tripod_removal}
\end{figure}

%We compute a 2-approximation $Z$ for this set cover instance and set
%\[
%K_{\tripod}:=\bigcup_{w\in (H\setminus Y)\setminus Z}{\cal P}(w).
%\]

% Now we embed $G''':=G''\setminus K_{\tripod}$ into a forest $F$ as follows. We set $F:=H\setminus (Y\cup Z)$ and then replace each $y\in H\setminus (Y\cup Z)$ with a star of size $|{\cal P}(y)|$.
% We keep the original edges of $H'$. 
% We show that the induced embedding of $G''$ into $F$ has bounded distortion. 
% Therefore, any bounded distortion embedding of $F$ into the line gives raise to a bounded distortion embedding of $G'''$ into the line.

%\begin{figure}[h]
%\def\svgscale{0.8}
%\centering
%   \input{Treetoline.pdf_tex}
%   \caption{Construction of the final embedding $f$ of a tree with no large tripods. Each yellow triangle represents a subtree $T_i$ rooted at $v_i$. The embedding of each $T_i$ follows the arrows.}
%\end{figure}

% We obtain such bounded embedding as follows. For every subtree $T\in F$ we compute the longest path $Q$ in $T$. For every $v_i\in V(Q)$ we consider the largest $T_i$ rooted at $v_i$ and otherwise disjoint from $Q$. Since $\Delta(G''')$ is bounded, $\Delta(F)$ hence $\Delta(T_i)$ are bounded as well. Moreover, $F$ hence $T$ does not contain large tripods. Thus, we obtain $|V(T_i)|\le \poly(c)$ for every $v_i\in V(Q)$. We compute an embedding $f_i$ for every $T_i$ via taking a walk obtained by doubling each $e\in E(T_i)$. Afterwards, we merge the embeddings leaving a $\poly(c)$ gap between $f_i(T_i)$ and $f_{i+1}(T_{i+1})$.

\textbf{Putting everything together.}
The final algorithm combines the above procedures for eliminating the three obstructions that we have identified.
At each obstruction elimination step, we remove a small set of vertices.
One additional complication is that, because $c$-embeddability into the line is not a hereditary property, this can produce a graph that does not admit a low-distortion embedding into the line.
We show that this issue can be avoided by deleting a slightly larger superset of vertices, which eliminates the obstruction at hand, while maintaining the existence of a low-distortion embedding.

\subsection{Organization}
The rest of the paper is organized as follows. 
We introduce necessary notation and definitions in Section~\ref{sect:Preliminaries}.
In Section~\ref{sect:main_algo}, we present our main algorithm and we state the main technical results needed.
In Section~\ref{sect:Rep_Lemma} we prove a technical lemma which will be applied throughout the paper. Sections~\ref{sect:Density_Reduction}, \ref{sect:Embed_into_Forest}, \ref{sect:Tripod_Elimin}, \ref{sect:Trees_no_large_Tripods} elaborate on the subroutines executed by the main algorithm.
%Finally, we prove the correctness of the main algorithm in Section~\ref{sect:main_algo}. 

\section{Preliminaries}\label{sect:Preliminaries}
We now introduce some notation and preliminary results that are used throughout the paper.

\subsection{Graphs}
%\textbf{Graphs.}
Given a graph $G$, we refer to its vertex set as $V(G)$ and to its edge set as $E(G)$.
%we refer to the set of vertices of $G$ as $V(G)$ and to the set of edges of $G$ as $E(G)$.
For any $C\subseteq V(G)$, we denote by $G[C]$ the subgraph of $G$ induced on $C$.
Let $d_G$ denote the shortest-path distance of $G$; unless otherwise noted, we assume that all edges in $G$ are undirected and have unit length.

\begin{definition}[Local density]
For any $v\in V(G)$ and $R\in \mathbb{N}$, we define
\[
\Delta_G(v,R)=\frac{|\ball_G(v,R)|-1}{2R}
\]
The \emph{local density} of the graph $G$ is defined to be
\[
\Delta(G)=\max_{v\in V(G), R\in \mathbb{N}}\Delta_G(v,R).
\]
\end{definition}

\begin{definition}[Tripod]
Let $G$ be a graph, $R\geq 1$, $v,v_1,v_2,v_3\in V(G)$, 
and let $P_1,P_2,P_3$ be paths in $G$,
where for all $i\in [3]$, $P_i$ is a path with endpoints $v$ and $v_i$.
Suppose that for all $i\neq j\in [3]$, and for all $u\in P_j$, we have $d_G(v_i, u)\geq R$.
In other words, each endpoint $v_i$ is at distance at least $R$ from every vertex in the other two paths.
Then we say that the tree $P_1\cup P_2\cup P_3$ is a \emph{$R$-tripod} with \emph{root $v$} (in $G$).
\end{definition}

\begin{figure}[h!]
\def\svgscale{0.6}
    \centering
%% Accompanies image file 'tripod.pdf' (pdf, eps, ps)
\begingroup%
  \makeatletter%
  \providecommand\color[2][]{%
    \errmessage{(Inkscape) Color is used for the text in Inkscape, but the package 'color.sty' is not loaded}%
    \renewcommand\color[2][]{}%
  }%
  \providecommand\transparent[1]{%
    \errmessage{(Inkscape) Transparency is used (non-zero) for the text in Inkscape, but the package 'transparent.sty' is not loaded}%
    \renewcommand\transparent[1]{}%
  }%
  \providecommand\rotatebox[2]{#2}%
  \newcommand*\fsize{\dimexpr\f@size pt\relax}%
  \newcommand*\lineheight[1]{\fontsize{\fsize}{#1\fsize}\selectfont}%
  \ifx\svgwidth\undefined%
    \setlength{\unitlength}{367.12384036bp}%
    \ifx\svgscale\undefined%
      \relax%
    \else%
      \setlength{\unitlength}{\unitlength * \real{\svgscale}}%
    \fi%
  \else%
    \setlength{\unitlength}{\svgwidth}%
  \fi%
  \global\let\svgwidth\undefined%
  \global\let\svgscale\undefined%
  \makeatother%
  \begin{picture}(1,0.24931565)%
    \lineheight{1}%
    \setlength\tabcolsep{0pt}%
    \put(0,0){\includegraphics[width=\unitlength,page=1]{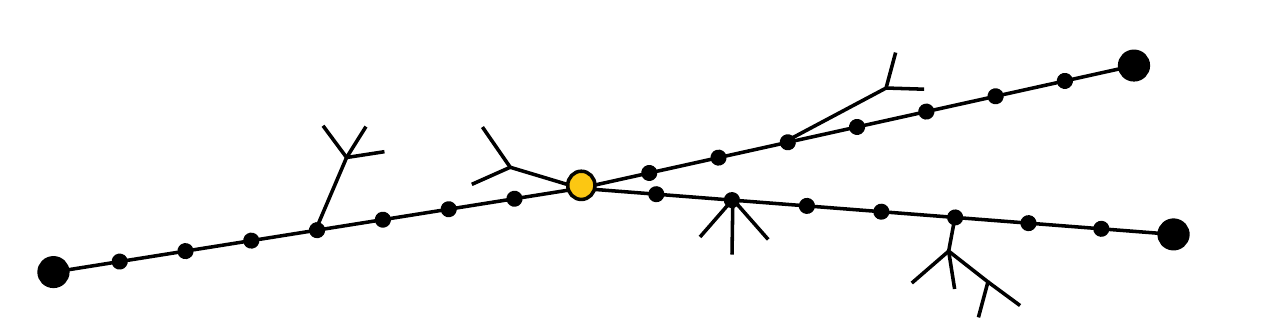}}%
    \put(0.40949744,0.05375928){\color[rgb]{0,0,0}\makebox(0,0)[lt]{\lineheight{1.25}\smash{\begin{tabular}[t]{l}$v$\end{tabular}}}}%
    \put(-0.00461504,0.05812865){\color[rgb]{0,0,0}\makebox(0,0)[lt]{\lineheight{1.25}\smash{\begin{tabular}[t]{l}$v_1$\end{tabular}}}}%
    \put(0.80575391,0.20711568){\color[rgb]{0,0,0}\makebox(0,0)[lt]{\lineheight{1.25}\smash{\begin{tabular}[t]{l}$v_2$\end{tabular}}}}%
    \put(0.83774047,0.08226497){\color[rgb]{0,0,0}\makebox(0,0)[lt]{\lineheight{1.25}\smash{\begin{tabular}[t]{l}$v_3$\end{tabular}}}}%
  \end{picture}%
\endgroup%

    \caption{A tripod rooted at $v$ with leaves $v_1,v_2,v_3$.}
    \label{fig:tripod}
\end{figure}
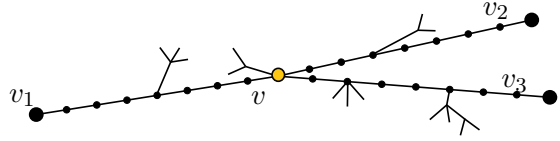

\subsection{Some useful approximation results}
For a graph $G$, a \emph{feedback vertex set} is some $X\subseteq V(G)$, such that $G\setminus X$ is acyclic.
In the Minimum Feedback Vertex Set problem we are given a graph $G$ and the goal is to find a feedback vertex set in $G$ of minimum cardinality.
We recall the following result on approximating the Minimum Feedback Vertex Set problem.

\begin{theorem}[Bafna \etal~\cite{bafna19992}]\label{thm:MFVS}
There exists a polynomial-time $2$-approximation algorithm for the Minimum Feedback Vertex Set problem.
\end{theorem}

Given a graph $G$ and some $\alpha\in [0,1)$, we say that some $X\subseteq V(G)$ is a \emph{$\alpha$-balanced vertex separator} (of $G$) if every connected component of $G\setminus X$ has at most $\alpha \cdot|V(G)|$ vertices.
We recall the following algorithmic result on computing balanced vertex separators.

\begin{theorem}[Feige \etal~\cite{feige2008improved}]\label{thm:v_separator_approx}
There exists a polynomial-time algorithm which given a graph that admits a $2/3$-balanced vertex separator of size $s$, outputs a $3/4$-balanced vertex separator of size at most $O(\sqrt{\log n} \cdot s)$.
\end{theorem}

Recall that an instance to the Minimum Set Cover problem consists of some set $U$ (the \emph{universe}), and a set ${\cal C}$ of subsets of $U$.
The goal is to find a subset of ${\cal C}$ of minimum cardinality that covers $U$.

\begin{theorem}[Chv\'{a}tal \cite{chvatal1979greedy}]\label{thm:SetCoverApprox}
There exists a polynomial-time $O(\log n)$-approximation algorithm for the Minimum Set Cover problem.
\end{theorem}

\subsection{Voronoi minors}\label{sect:Voronoi}

For some metric space $\metspace=(X,\rho)$, and some $R>0$, we say that some $N\subseteq X$ is a \emph{$R$-net} of $M$ if
for any $p,q\in R$, $\rho(p,q)>R$, and $X\subseteq \bigcup_{p\in N} \ball_{\metspace}(p,R)$.
For a graph $G$, we say that some $N\subseteq V(G)$ is a $R$-net of $G$ if $N$ is a $R$-net of the shortest-path metric of $G$.

\begin{definition}[Graphical Voronoi partition]
Let $G$ be a graph, and let $Y\subseteq V(G)$.
Let ${\cal P}$ be a partition of $V(G)$ satisfying the following conditions:
\begin{description}
\item{(1)}
Every cluster in ${\cal P}$ contains exactly one vertex in $Y$.
\item{(2)}
For any $v\in V(G)$, the cluster containing $v$, ${\cal P}(v)$, also contains some nearest neighbor of $v$ in $Y$.
\item{(3)}
For any cluster $C\in {\cal P}$, we have that $G[C]$ is connected.
\end{description}
We say that ${\cal P}$ is a \emph{Voronoi partition} of $G$ \emph{centered at $Y$}.
\end{definition}

We note the following easy fact.

\begin{lemma}
For any graph $G$, and $Y\subseteq V(G)$, there exists a Voronoi partition ${\cal P}$ of $G$ centered at $Y$.
\end{lemma}

\begin{proof}
Construct ${\cal P}$ by assigning each $v\in V(G)$ to the cluster containing its nearest neighbor in $Y$.
In order to ensure that each cluster $C$ induces connected subgraph $G[C]$ it suffices to ensure that shortest-paths in $G$ are unique.
This can be achieved by breaking ties between different paths lexicographically (viewing paths as sequences of vertices with unique integer labels) (see also \cite{badoiu2005approximation}).
\end{proof}

\begin{definition}[$R$-Minor]\label{def:r-minor}
Let $G$ be a graph, $R>0$, and let
$N$ be a $R$-net of $G$.
Let ${\cal P}$ be a Voronoi partition of $G$ centered at $N$.
Let $H$ be the minor of $G$ obtained by contracting each cluster in $C$ in ${\cal P}$ into the unique net point in $C$.
Then we say that ${\cal P}$ is a \emph{$R$-partition} and $H$ is a \emph{$R$-minor} of $G$ induced by ${\cal P}$
 (see Figure \ref{fig:voronoi_minor} for an example).
\end{definition}

\begin{figure}
    \def\svgscale{0.5}
    \centering
%% Accompanies image file 'Voronoi_minor.pdf' (pdf, eps, ps)
\begingroup%
  \makeatletter%
  \providecommand\color[2][]{%
    \errmessage{(Inkscape) Color is used for the text in Inkscape, but the package 'color.sty' is not loaded}%
    \renewcommand\color[2][]{}%
  }%
  \providecommand\transparent[1]{%
    \errmessage{(Inkscape) Transparency is used (non-zero) for the text in Inkscape, but the package 'transparent.sty' is not loaded}%
    \renewcommand\transparent[1]{}%
  }%
  \providecommand\rotatebox[2]{#2}%
  \newcommand*\fsize{\dimexpr\f@size pt\relax}%
  \newcommand*\lineheight[1]{\fontsize{\fsize}{#1\fsize}\selectfont}%
  \ifx\svgwidth\undefined%
    \setlength{\unitlength}{382.90661893bp}%
    \ifx\svgscale\undefined%
      \relax%
    \else%
      \setlength{\unitlength}{\unitlength * \real{\svgscale}}%
    \fi%
  \else%
    \setlength{\unitlength}{\svgwidth}%
  \fi%
  \global\let\svgwidth\undefined%
  \global\let\svgscale\undefined%
  \makeatother%
  \begin{picture}(1,0.42540436)%
    \lineheight{1}%
    \setlength\tabcolsep{0pt}%
    \put(0,0){\includegraphics[width=\unitlength,page=1]{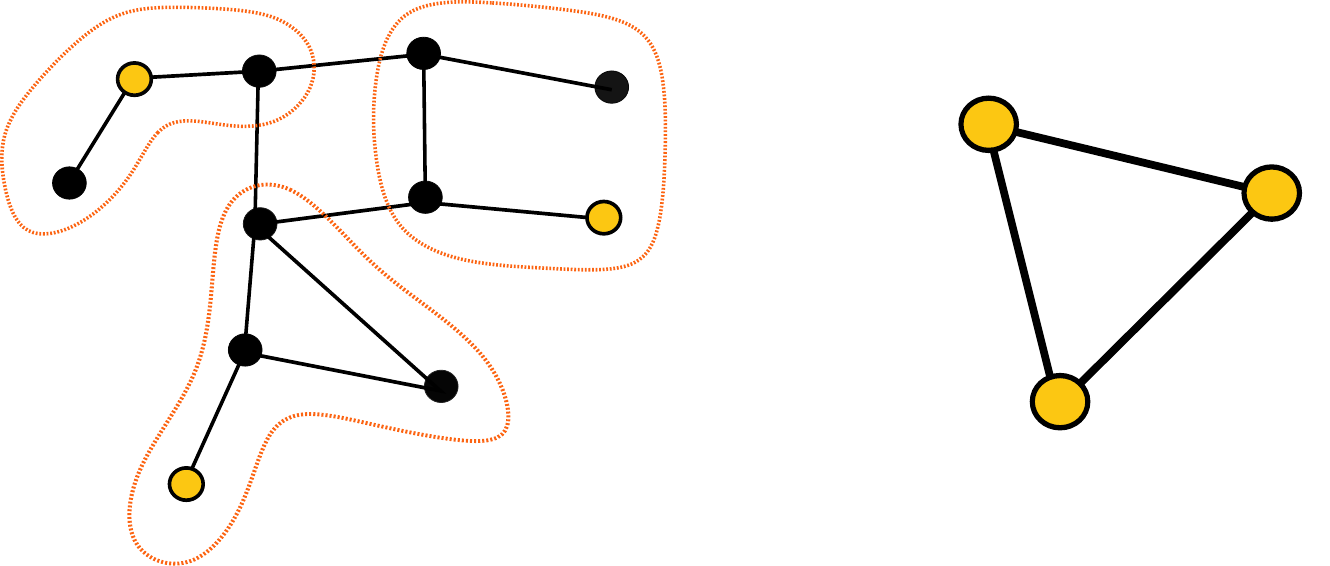}}%
    \put(0.71475472,0.36717378){\color[rgb]{0,0,0}\makebox(0,0)[lt]{\lineheight{1.25}\smash{\begin{tabular}[t]{l}$y_1$\end{tabular}}}}%
    \put(0.10969872,0.02909937){\color[rgb]{0,0,0}\makebox(0,0)[lt]{\lineheight{1.25}\smash{\begin{tabular}[t]{l}$y_2$\end{tabular}}}}%
    \put(0.8920095,0.31671769){\color[rgb]{0,0,0}\makebox(0,0)[lt]{\lineheight{1.25}\smash{\begin{tabular}[t]{l}$y_3$\end{tabular}}}}%
    \put(0.06228904,0.38937111){\color[rgb]{0,0,0}\makebox(0,0)[lt]{\lineheight{1.25}\smash{\begin{tabular}[t]{l}$y_1$\end{tabular}}}}%
    \put(0.76078771,0.07387232){\color[rgb]{0,0,0}\makebox(0,0)[lt]{\lineheight{1.25}\smash{\begin{tabular}[t]{l}$y_2$\end{tabular}}}}%
    \put(0.40536528,0.29239153){\color[rgb]{0,0,0}\makebox(0,0)[lt]{\lineheight{1.25}\smash{\begin{tabular}[t]{l}$y_3$\end{tabular}}}}%
    \put(0,0){\includegraphics[width=\unitlength,page=2]{Voronoi_minor.pdf}}%
  \end{picture}%
\endgroup%

         \caption{A Voronoi partition centred at $3$-net $N=\{y_1,y_2,y_3\}$ and a corresponding 3-minor.}
         \label{fig:voronoi_minor}
\end{figure}

\section{The Main Algorithm}\label{sect:main_algo}

In this Section we present and analyze the main algorithm of this paper.
For the sake of clarity, we first state some key technical ingredients used by the algorithm.
We then present the main algorithm and its analysis.
The proofs of the technical ingredients are deferred to latter Sections.

%For sake of clarify, we first state some key technical ingredients that are used by the algorithm, and we defer their proofs to latter sections.

%The problem our algorithm solves can be formulated as follows:
%\begin{problem}
%Given an undirected unweighted graph $G$ and some $c\ge 1, k\ge 0$ either correctly decide that $G$ does not admit a $(k,c)$-embedding into the line or in polynomial time compute a $(\poly(c,\log(n))k,\poly(c))$- embedding of $G$ into the line.
%\end{problem}

\subsection{Technical ingredients used by the main algorithm}

\textbf{Density reduction.}
The first technical ingredient used by the main algorithm is a procedure for reducing the local density of the input graph.
This is summarized in Lemma \ref{lem:density_reduction}.
Its proof is given in Section \ref{sect:Density_Reduction}.

\begin{lemma}[Density Reduction]\label{lem:density_reduction}
There exists a polynomial-time algorithm given given a graph $G$, $k\geq 0$, $c\geq 1$, terminates with exactly one of the following outcomes:
\begin{description}
\item{(1)}
Correctly decides that $G$ does not admit a $(k,c)$-embedding into the line.

\item{(2)}
Outputs some $Y \subseteq V(G)$  such that
$\Delta(G\setminus Y)\leq c$,
with 
$\left|Y\right|=O(c k \log^{3/2} n)$.
In particular, if $\Delta(G)\le c$, then the algorithm outputs $\emptyset$.
\end{description}
\end{lemma}

\textbf{Eliminating large metrical cycles.}
The next technical ingredient is a procedure for eliminating large metrical cycles.
This is summarized in Lemma \ref{lem:forest}, whose proof is given in Section \ref{sect:Embed_into_Forest}.

%The main statement contained in Section~\ref{sect:Embed_into_Forest}.

\begin{lemma}[Embedding into a forest]\label{lem:forest}
There exists a polynomial-time algorithm which given a graph $G$, $c\geq 1$, and $k\geq 0$, terminates with exactly one of the following outcomes:
\begin{description}
\item{(1)}
Correctly decides that $G$ does not admit a $(k,c)$-embedding into the line.
\item{(2)}
Outputs a $c$-net $N$ of $G$, a $c$-partition ${\cal P}$ centered at $N$, a $c$-minor $H$ induced by ${\cal P}$, and some feedback vertex set $X$ of $H$, with $|X|\leq 2k$.
\end{description}
\end{lemma}

%Now $H$ is a forest, the only obstruction left is large metric tripods. 

\textbf{Eliminating large metrical tripods.}
The next obstruction that the main algorithm needs to remove is large metrical tripods.
This is done using Lemmas \ref{lem:tripod_distortion} and \ref{lem:tripod_eliminaton}.
Their proofs appear in Section \ref{sect:Tripod_Elimin}.

\begin{lemma}[Tripods as obstructions to embeddability]\label{lem:tripod_distortion}
Let $G$ be a graph, $R\geq 1$, and let $J$ be a $R$-tripod in $G$.
Then for any $c$-embedding of $G$ into the line we have $c\geq 2R$.
\end{lemma}

\begin{lemma}[Tripod elimination]\label{lem:tripod_eliminaton}
There exists a polynomial-time algorithm which given a forest $F$, $R\geq 1$, $k\geq 0$, terminates with exactly one of the following outcomes:
\begin{description}
\item{(1)}
Correctly decides that there exists no $X'\subseteq V(F)$, with $|X'| \leq k$, such that $F\setminus X'$ does not contain any $R$-tripod as a subgraph.

\item{(2)}
Outputs some $X'\subseteq V(F)$, with $|X'| = O(k \log n)$, such that $F\setminus X'$ does not contain any $R$-tripod as a subgraph.
\end{description}
\end{lemma}

\textbf{Embedding a tree with no large tripods into the line.}
Once all the obstructions have been removed, the problem is reduced to computing an embedding of a tree with no large tripods into the line.
This is done using Lemma \ref{lem:embed_trees_no_tripods}, whose proof appears in Section \ref{sect:Trees_no_large_Tripods}.

%At this stage, a graph looks like a sparse forest with no large tripod; we embed each component of the forest independently as follows. 

\begin{lemma}\label{lem:embed_trees_no_tripods}
Let $R\geq 1$, and let $T$ be a tree that does not contain any $R$-tripod as a subgraph.
Then $T$ admits a $O(\Delta(T)\cdot R)$-embedding into the line.
Moreover, this embedding can be computed in polynomial time.
\end{lemma}

\textbf{The Repairing Lemma.}
The main algorithm proceeds in several steps.
At each step, it uses some of the  procedures described above to delete small subsets of vertices.
However, because $c$-embeddability into the line is not a hereditary property, it is possible that the deletion of some small set of vertices destroys some candidate solution.
As an illustrative example, let $G$ be the $3\times (n/3)$ grid.
Note that $G$ admits a $O(1)$-embedding into the line (i.e.~without outliers).
This embedding can be realized by consecutively traversing the columns of the grid.
Let $U$ be the set of vertices that do not lie on the outer boundary cycle of $G$.
Then, $G\setminus U$ is the $(2n/3+2)$-cycle, and therefore any embedding of $G\setminus U$ into the line has distortion $\Omega(n)$. 
However, by removing one additional vertex from $G\setminus U$ we obtain a path, which admits a $1$-embedding into the line (see Figure \ref{fig:grid_to_cycle}).
%Thus, in order to keep the distortion small, some additional vertices of $G\setminus U$ must be removed.
We show that the above ``repairing'' process can be performed for arbitrary $U$.
Lemma \ref{lem:repairing_lemma} summarizes this result.
Its proof is given in Section \ref{sect:Rep_Lemma}.

\begin{figure}[h!]
    \centering
    \includegraphics[height=1.2cm]{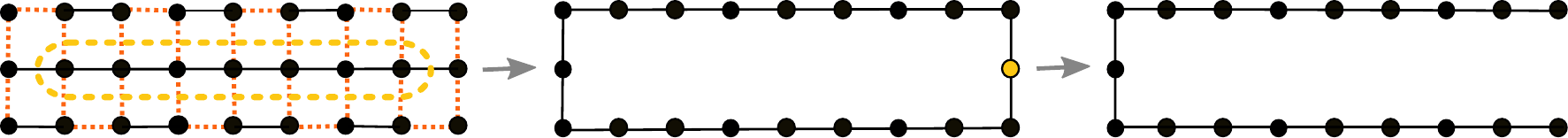}
    \caption{A $3\times n$ grid $G$ can be embedded into the line with distortion $O(1)$; one could follow the red dotted path on the grid an embed the vertices consequently. A yellow line depicts $U$. Now, if we delete a yellow vertex from $G\setminus U$, the resulting graph will be just a path.}
    \label{fig:grid_to_cycle}
\end{figure}

\begin{lemma}[Repairing Lemma]\label{lem:repairing_lemma}
Let $G$ be a graph, $U\subset V(G)$,  $k\geq 0$, $c\geq 1$.
Suppose that $G$ admits a $(k,c)$-embedding into a line.
Then, $G\setminus U$ admits a $((2c+1)|U|+k, 4c^3+c)$-embedding into a line.
\end{lemma}

\subsection{The algorithm}
Given the technical ingredients presented above, we are now ready to describe our main algorithm.
Recall that the input consists of a graph $G$, and $k\geq 0$, $c\geq 1$.
The algorithm proceeds in the following steps.

\begin{description}
\item{\textbf{Step 1: Density reduction.}}
Using the algorithm from Lemma \ref{lem:density_reduction} we can either correctly decide that $G$ does not admit a $(k,c)$-embedding into the line, in which case we terminate, or we compute some $X_{\density}\subseteq V(G)$, with $|X_{\density}|\leq O(ck \log^{3/2} n)$, such that $\Delta(G\setminus X_{\density})\leq c$.

\item{\textbf{Step 2: Cycle elimination.}}
Let $k'=(2c+1)|X_{\density}|+k$ and $c'=4c^3+c$.
Using the algorithm from Lemma \ref{lem:forest} we either correctly decide that $G'$ does not admits a $(k',c')$-embedding into the line, or we compute a $c'$-net $N$ of $G'$, a $c'$-partition ${\cal P}$ centered at $N$, a $c'$-minor $H$ induced by ${\cal P}$, and some feedback vertex set $Y_{\forest}$ of $H$, with $|Y_{\forest}|\leq 2k'$.
If $G'$ does not admit a $(k',c')$-embedding into the line, then we terminate by deciding that $G$ does not admit a $(k,c)$-embedding into the line.

\item{\textbf{Step 3: Tripod elimination.}}
Let $F=H\setminus Y_{\forest}$, and recall that $Y_{\forest}$ is a feedback vertex set for $H$, and thus $F$ is a forest.
Using the algorithm from Lemma  \ref{lem:tripod_eliminaton}, in polynomial time, we either decide that there exists no $Y_{\tripod}\subseteq V(F)$, with $|Y_{\tripod}|\leq k'$, such that $F\setminus T_{\tripod}$ does not contain any $(c'/2+1)$-tripod, in which case we terminate deciding that $G$ does not admit a $(k,c)$-embedding into the line, or we compute some $Y_{\tripod}\subseteq V(F)$, with $|Y_{\tripod}|=O(k \log n)$, such that $F\setminus Y_{\tripod}$ does not contain any $(c'/2+1)$-tripods.

\item{\textbf{Step 4: Embedding into a forest.}}
Let $F'=F\setminus Y_{\tripod}$.
Let
\[
X_{\forest} = \bigcup_{v\in Y_{\forest}} {\cal P}(v),
\]
\[
X_{\tripod} = \bigcup_{v\in Y_{\tripod}} {\cal P}(v),
\]
and
\[
K = X_{\density} \cup X_{\forest} \cup X_{\tripod}.
\]
Let $F''$ be the forest obtained from $F'$ as follows.
Initially, we set $F'':=F'$.
For each $v\in V(G)\setminus K$, let $u(v)$ be the unique vertex in $N\cap {\cal P}(v)$; we add $v$ to $F''$ as a leaf attached to $u(v)$.
This completes the construction of the forest $F''$.

\item{\textbf{Step 5: Embedding into the line.}}
Finally, we compute an embedding $f$ of $F''$ into the line using the algorithm from Theorem \ref{lem:embed_trees_no_tripods}.
We output the embedding $\varphi :=2c'c\cdot f$ (that is, $f$ scalled by a factor of $2c'c$).
\end{description}

This completes the description of the main algorithm.

\subsection{Analysis of the main algorithm}
We now analyze the main algorithm presented above.
First, we state some auxiliary properties of $c$-minors and $c$-partitions.
Their proofs appear in Section \ref{sec:minors}.

\begin{lemma}\label{lem:properties_of_voron}    
Let $G$ be a graph, $R\ge 1$. Let $N$ be a $R$-net of $G$, ${\cal P}$ a corresponding $R$-partition and $H$ a $R$-minor  $G$ induced by ${\cal P}$. Then for any $Y\subseteq V(H)$ all of the following hold:
\begin{description}
    \item{(1)} $N':=N\setminus Y$ is a $R$-net in $G':=G\setminus \left(\cup_{v\in Y}{\cal P}(v)\right)$
    \item{(2)} ${\cal P}':={\cal P}\setminus \left(\cup_{v\in Y}\{{\cal P}(v)\}\right)$ is the $R$-partition of $G'$ centered at $N'$
    \item{(3)} $H':=H\setminus Y$ is the $R$-minor of $G'$ induced by ${\cal P}'$.
\end{description}
\end{lemma}

\begin{lemma}\label{lem:distance_in_graph_vs_minor}
Let $G$ be a graph and let $R>0$. Let $N$ be $c$-net of $G$, ${\cal P}$ a $c$-partition centered at $N$, and $H$ a $R$-minor induced by ${\cal P}$.
Then for any $u, v\in N$ we have
$d_{H}(u,v)\le d_G(u,v)$.
\end{lemma}

We now have all the necessary ingredients in place to prove Theorem \ref{thm:main_thm}, which is the main result of this paper.

\begin{proof}[Proof of Theorem \ref{thm:main_thm}]
We analyze the algorithm presented above.
By Lemma \ref{lem:density_reduction}, if we terminate at Step 1, then we correctly decide that $G$ does not admit a $(k,c)$-embedding.
Otherwise, by Lemma \ref{lem:repairing_lemma}, it follows that if $G$ admits a $(k,c)$-embedding into the line, then $G'=G\setminus X_{\density}$ admits a $(k',c')$-embedding into the line,
with $k'=(2c+1)|X_{\density}|+k = O(c^2 k \log^{3/2}n))$ and $c'=4c^3+c$.

By Lemma \ref{lem:forest}, if we decide that $G'$ does not admit a $(k',c')$-embedding into the line, then, by the above discussion, this certifies that $G$ does not admit a $(k,c)$-embedding into the line; we can thus correctly decide this fact in Step 2.

Suppose that $G'$ admits a $(k',c')$-embedding into the line.
Thus, there exists some $K'\subseteq V(G')$, with $|K'|\leq k'$, such that $G'\setminus K'$ admits a $c'$-embedding into the line.
Let $J$ be the set of all $v\in N$ such that the Voronoi cell of $v$ intersects $K'$, that is 
\begin{align*}
J &= \{v\in N : K' \cap {\cal P}(v) \neq \emptyset\}.
\end{align*}
We claim that $F\setminus J$ does not contain any $(3c'/2+1)$-tripod.
For the sake of contradiction, suppose that $F\setminus J$ contains some $(3c'/2+1)$-tripod $T=P_1\cup P_2\cup P_3$, where $P_1,P_2,P_3$ are three paths sharing a root $r$.
For any $i\in [3]$ let $z_i$ be the endpoint of $P_i$ other than $r$.
Then for any $i\in[3]$ there exists a path $Q_i$ in $G'\setminus K'$ between $r$ and $z_i$. We claim that for all $i\neq j\in [3]$, for all $u\in V(Q_j)$, we have
$d_{G\setminus K}(z_i,u)\ge c'/2+1$.
By Lemma~\ref{lem:properties_of_voron}, $F\setminus J$ is a $c'$-minor of $G'\setminus K'$ with respect to the Voronoi partition ${\cal P}_J$ with ${\cal P}(w)={\cal P}_J(w)$ for all $w\in V(F\setminus J)$. Let $w'$ be such that $u\in {\cal P}_J(w')$. By Lemma~\ref{lem:distance_in_graph_vs_minor} obtain
\begin{align*}
    d_{G'\setminus K'}(z_i,u)&\ge d_{G'\setminus K'}(z_i,w')-d_{G'\setminus K'}(w',u) \tag{\text{by the triangle inequality}}\\
    &\ge d_{G'\setminus K'}(z_i,w')-c' \tag{since $u\in {\cal P}_J(w')$}\\
    &\ge d_{F\setminus J}(z_i,w)-c' \tag{by Lemma \ref{lem:distance_in_graph_vs_minor}}\\
    &\ge 3c'/2+1-c' \tag{since $T$ is a $(3c'/2+1)$-tripod}\\
    &=c'/2+1.
\end{align*}

Therefore, by Lemma \ref{lem:tripod_distortion} we conclude that $G'\setminus K'$ does not admit a $c'$-embedding into the line, which is a contradiction.
Therefore, we have established that if $G'$ admits a $(k,c)$-embedding into the line, then there exists some $J\setminus V(F)$, with $|J|\leq k'$, such that $F\setminus J$ does not contain any $(3c'/2+1)$-tripods.

Therefore, in Step 3, if we do not find a set $Y_{\tripod}$ of the desired size, then we correctly decide that $G$ does not admit a $(k,c)$-embedding into the line.

Next consider the case where in Step 3 we compute a set $Y_{\tripod}$ of the desired size.
Since $F'$ does not contain any $(3c'/2+1)$-tripods, it follows by the construction of $F''$, that $F''$ does not contain any $(3c'/2+3)$-tripods (since every leaf in $F$ becomes the center of a star in $F'$).
Moreover, we have $\Delta(F'') \leq \Delta(F') \cdot O(c' \Delta(G'))$, since every vertex in $F''$ corresponds to a star that contains the vertices of a Voronoi cell in $G'$, and every such cell has size at most $O(c' \Delta(G'))$.
Thus, by Lemma \ref{lem:embed_trees_no_tripods} we compute a $c''$-embedding of $F''$ into the line, where
$c''=O(\Delta(F'') c') = O(\Delta(F') c^3 \Delta(G')) = O(\Delta(F) c^3 \Delta(G)) = O(\Delta(H) c^4)$, since $\Delta(\Gamma_1)\leq \Delta(\Gamma_2)$ for all $\Gamma_1\subset \Gamma_2$.
Moreover we have $\Delta(H)\leq \Delta(G') \cdot O(c' \cdot \Delta(G')) = O(c^5)$, since every vertex in $H$ corresponds to a Voronoi cell consisting of at most $O(c'\cdot \Delta(G'))$ vertices.
Therefore $c''=O(c^9)$, and thus we have obtained a $O(c^9)$-embedding $f$ of $F''$ into the line. Note that since $V(F'')=V(G\setminus K),$ it follows that $f $ is also a $(\kappa,\sigma)$-embedding of $G$ into the line, where $\kappa=|K|$, for some  $\sigma\geq 1$.

It remains to bound $\kappa$ and $\sigma.$ 
We have
\[
\kappa = |X_{\density}|+|X_{\forest}|+|X_{\tripod}|.
\]
Since $G$ admits a $(k,c)$-embedding into the line, it follows from Lemma~\ref{lem:density_reduction} that 
\[
|X_{\density}|=O(ck\log^{3/2}{n}).
\] 
Moreover, $\Delta(G\setminus X_{\density})\le c$, thus 
for any $\widetilde{c}$-partition ${\cal P}$ induced by an arbitrary $\widetilde{c}$-net $N$ of $G\setminus X_{\density}$, and any $v\in N$, we have
\[
|{\cal P}(v)|= O(\widetilde{c}\cdot \Delta(G\setminus X_{\density}))= O(\widetilde{c}\cdot c).
\]
Therefore, using Lemma~\ref{lem:forest} with $\widetilde{c}:=c'$ in the Step 3 we obtain
\begin{align*}
|X_{\forest}|&=O(c'\cdot c)\cdot 2 k'\\
&=O((4c^3+c)\cdot c\cdot(c^2k\log^{3/2}{n}))\\
&=O(c^6k\log^{3/2}{n}).
\end{align*}
Similarly, from Lemma~\ref{lem:tripod_eliminaton}, we have
\begin{align*}
|X_{\tripod}|&=O(c'\cdot c)O(k'\log{n})\\
&=O((4c^3+c)\cdot c)\cdot O(c^2k\log^{3/2}{n})\log{n})\\
&=O(c^6k\log^{5/2}{n}),
\end{align*}
which implies that
\begin{align*}
\kappa&=O(ck\log^{3/2}{n})+O(c^6k\log^{3/2}{n})+O(c^6k\log^{5/2}{n})\\
&=O(c^6k\log^{5/2}{n}).
\end{align*}
To find $\sigma,$ we show that $G\setminus K$ admits a $O(c^4)$-embedding $\iota$ into $F''$ with $\iota(v)=v$ for all $v\in G\setminus K.$ 
By Lemma~\ref{lem:properties_of_voron} $F'$ is a $c'$-minor of $G'\setminus(X_{\forest}\cup X_{\tripod})=G\setminus K$ with respect to the partition ${\cal P}':={\cal P}\setminus(\cup_{v\in Y_{\forest}\cup Y_{\tripod}}{\cal P}(v)).$ Consider arbitrary $x_1,x_2\in V(G\setminus K)$ and let  $v_1,v_2\in V(F')$ be such that $x_1\in {\cal P}'(v_1)$, $x_2\in {\cal P}'(v_2)$. Let $Q$ be the unique $v_1$-$v_2$ path in $F'$. We use $Q$ to construct a $v_1$-$v_2$ path $P$ in $G\setminus K$, with 
\[
\length(Q)\le\length(P)\le 2c'c\cdot \length(Q).
\] 
Since $F'$ is a $c'$-minor of $G\setminus K$, for any $\{w_1,w_2\}\in E(Q)$ there is $\{z_1,z_2\}\in E(G\setminus K)$ with $z_i\in {\cal P'}(w_i)$ for $i\in [2]$. Moreover, for any $w\in V(Q)$ the corresponding ${\cal P}'(w)$ is a connected subgraph such that $|V({\cal P'}(w_i))|\le 2c'\Delta(G\setminus K)+1=2c'c+1$. Thus, $Q$ induces a walk $W\subseteq G\setminus K$ with $|V(W)|\le 2c'c\cdot \length(Q)$ and $v_1,v_2\in W$. It follows that there is a $v_1$-$v_2$ path $P$ in $W$, such that 
\[
\length(P)\le 2c'c\cdot \length(Q).
\]
Note that since $Q$ is the $v_1$-$v_2$ shortest path in $F'$, we obtain
\[
\length(P)\le 2c'c\cdot d_{F'}(v_1,v_2)= 2c'c\cdot d_{F''}(v_1,v_2),
\]
where the last equality follows from the construction of $F''$.

We claim that $\iota$ has contraction $O(c^4)$. By construction of $F''$ we have that $d_{F''}(x_i,v_i)=1$ thus 
\begin{align*}
    d_{G\setminus K}(x_i,v_i)\le
    c'\le c'd_{F''}(x_i,v_i).
\end{align*}
Therefore, we have that
\begin{align*}
    d_{G\setminus K}(x_1,x_2)&\le d_{G\setminus K}(x_1,v_1)+d_{G\setminus K}(v_1,v_2)+d_{G\setminus K}(v_2,x_2)\\
    &\le c'd_{F''}(x_1,v_1) + 2c'c\cdot \length(Q) + c'd_{F''}(v_2,x_2)\\
    &\le 2c'c \cdot d_{F''}(x_1,v_1) + 2c'c\cdot d_{F''}(v_1,v_2) + 2c'c\cdot d_{F''}(v_2,x_2).
\end{align*}
 Since $F''$ is a tree, it follows that
\[
2c'c \cdot d_{F''}(x_1,v_1) + 2c'c\cdot d_{F''}(v_1,v_2) + 2 c'c\cdot d_{F''}(v_2,x_2)=2c'c \cdot d_{F''}(x_1,x_2).
\]
Since $c'=O(c^3)$, it follows that the contraction of $\iota$ is at most $O(c^4)$.
Now we prove that the expansion of $\iota$ is $O(1)$. We claim that 
\begin{align*}
    d_{F''}(x_1,x_2)\le d_{G\setminus K}(x_1,x_2)+2.
\end{align*}
By the construction of $F''$  we have 
\begin{align*}
    d_{F''}(x_1,x_2)&=d_{F''}(x_1,v_1)+d_{F''}(v_1,v_2)+d_{F''}(v_2,x_2)\\
    &=d_{F''}(x_1,v_1)+d_{F'}(v_1,v_2)+d_{F''}(v_2,x_2)\\
    &=d_{F'}(v_1,v_2)+2.
\end{align*}
Since $F'$ is a $c'$-minor of $G\setminus K$, by Lemma~\ref{lem:distance_in_graph_vs_minor} we get
\begin{align*}
    d_{F'}(v_1,v_2)+2\le d_{G\setminus K}(v_1,v_2)+2,
\end{align*}
which proves that the expansion of $\iota$ is $O(1)$.
We thus obtain that the distortion of $\iota$ is $O(c^4)$.

Therefore, we obtain that the map $\phi:=f\circ\iota: G\setminus K\to \reals^1$ has distortion $\sigma=O(c^9)\cdot O(c^4)=O(c^{13})$, which concludes the proof.
\end{proof}

\section{Proof of the Repairing Lemma}\label{sect:Rep_Lemma}

This Section is devoted to proving Lemma \ref{lem:repairing_lemma}.
 First, we prove two auxiliary statements.
\begin{lemma}\label{moving lemma}
Let $G$ be a graph, $k>0$, $c>1$. Assume that $G$ admits a $(k,c)$-embedding into a line. Suppose $G$ admits a $(k,c)$-embedding into the line realized by $f:G\setminus K\to \reals$. Then, there exists a $(k,c)$-embedding $f'$ of $G$ into a line such that if $j>i$ then for any $v\in G_i,w\in G_j$ we have $f'(w)>f'(v)$.
\end{lemma}
\begin{proof}
Let
\begin{align*}
    v_1=\argmin_{v\in V(G)\setminus K}\{f(v)\}\\
    v_2=\argmax_{v\in V(G)\setminus K}\{f(v)\}
\end{align*}
and let $M=f(v_2)-f(v_1)$. Without loss of generality, we can assume that $f(v_1)=0$ and $f(v_2)=M$ by setting $f(v):=f(v)-f(v_1)$. For each $v\in G_i$ we define
\[
f'(v)=f(v)+2i\cdot M.
\]
We claim that $f'$ and $f$ have the same distortion. If $v,w \in G_i$ then we have 
\[
|f'(w)-f'(v)|=|(f(w)+2i\cdot M)-(f(w)+2i\cdot M)|=|f(w)-f(v)|.
\] 
If $v\in G_i$ and $w\in G_j$ for $i\neq j$ then the distance between them in the embedding does not contribute to the distortion.

It remains to show that $f'(w)>f'(v)$ for all $w\in G_j,v\in G_i$ with $j>i$. We have
\begin{align*}
    f'(w)-f'(v)&=f(w)-f(v)+2(j-i)M >-M+2M >0
\end{align*}
and the claim follows by induction.
\end{proof}
Let $G=(V(G),E(G))$ be a graph and let $f:G\to \reals$. Consider $Z=\{v_1,\dots, v_m\}\subseteq V(G)$ such that $f(v_1)<f(v_2)<\dots <f(v_m)$. Then $Z$ is \emph{consecutive with respect to $f$} if for all $w\in V(G)\setminus U$ either $f(w)<f(v_1)$ or $f(v_m)<f(w)$.
\begin{lemma}\label{C-cut}
Let $G$ be a graph, $c>0$. Assume that $G$ admits a $(0,c)$-embedding into the line realized by $f:G\to\reals$. Let $Z =\{z_1,\dots,z_m\}\subseteq V$ be consecutive with respect to $f$. Suppose that $f(v_m)-f(v_1)\ge c$; then $Z$ is a vertex separator in $G.$
\end{lemma}
\begin{proof}
We claim that
\begin{align*}
    X=\{x\in V(G)\,|\,f(x)<f(v_1)\}\\
    Y=\{y\in V(G)\,|\,f(v_k)<f(y)\}
\end{align*} 
are disconnected in $G\setminus Z$. Assume otherwise; then there exists $\{x,y\}\in E(G)$ with $x\in X,y\in Y.$ Thus 
\begin{align*}
    |f(y)-f(x)|&=f(y)-f(v_k)+f(v_k)-f(v_1)+f(v_1)-f(x)\\
    &\ge c+1\\
    &>c\cdot d_{G}(x,y)
\end{align*}
which contradicts the distortion assumption.
\end{proof}

We can now prove the Repairing Lemma:

\begin{proof}[Proof of Lemma~\ref{lem:repairing_lemma}]
Let $f$ be a $(K,c)$-embedding of $G$ into the line, with $|K|=k$.
Let $U' = U\cap K$, and $U''=U\setminus K$.
For any $v\in U\setminus K$, let
\begin{align*}
I_{\inner}(v) &= \ball_{\mathbb{R}}(f(v),c),\\
I_{\outerr}(v) &=  \ball_{\mathbb{R}}(f(v), 2c^2)\setminus I_{\inner}(v),\\
V_{\inner}(v) &= \{u\in V(G)\setminus K:f(u)\in I_{\inner}(v)\}\\
V_{\outerr}(v) &= \{u\in V(G)\setminus K:f(u)\in I_{\outerr}(v)\}
\end{align*}

\begin{figure}[h]
\def\svgscale{0.9}
    \centering
     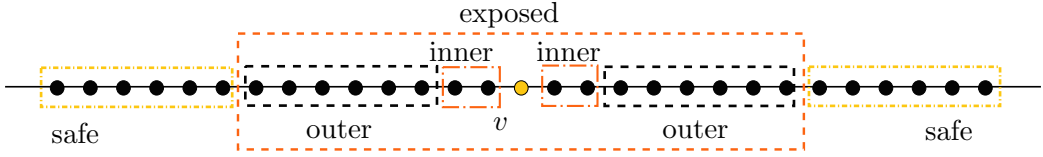
    \caption{Inner, outer, safe and exposed vertices with respect to $v$ for $c=2$.}
    \label{fig:my_label}
\end{figure}

Let also
\begin{align*}
I_{\inner} &= \bigcup_{v\in U\setminus K} I_{\inner}(v)\\
I_{\outerr} &= \left(\bigcup_{v\in U\setminus K} I_{\outerr}(v)\right) \setminus I_{\inner}\\
V_{\inner} &= \{ u\in V(G)\setminus K : f(u) \in I_{\inner}  \},\\
V_{\outerr} &= \{ u\in V(G)\setminus K : f(u) \in I_{\outerr}  \},\\
V_{\exposed} &= V_{\inner} \cup V_{\outerr},\\
V_{\safe} &= V(G) \setminus V_{\exposed}.
\end{align*}
We can now define
\[
K' = K\cup V_{\inner}.
\]
Since the minimum distance in $G$ is one, and $f$ is non-contracting, it follows that 
\[
|K'| \leq |K| + (2c+1) |U|.
\]

Let $c'=(4c^3+c)$.
It remains to construct any $(K',c')$-embedding $f'$.
By lemma \ref{moving lemma}
it is enough to construct a $c'$-embedding for each connected component of $G\setminus K'$

We may thus focus on any connected component $C$ of $G\setminus K'$.
Let $f'=(4c^2+1) \cdot f|_C$ (that is, $f'$ is the restriction of $f$ on $C$ scaled by a factor of $4c^2+1$).
It suffices to show that  $f'$ is a $(4c^3+c)$-embedding of $C$.

If there exist $v\in U\setminus K$, and $u\in C$ such that $f(v)<f(u)$, then we set
\begin{align*}
z_L &= \argmax_{v\in U\setminus K : \forall u\in C, f(z_L)<f(u)} \{f(v)\},
\end{align*}
Similarly, if there exist $v\in U\setminus K$, and $u\in C$ such that $f(v)>f(u)$, then we set
\begin{align*}
z_R &= \argmin_{v\in U\setminus K : \forall u\in C, f(z_R)>f(u)} \{f(v)\}.
\end{align*}

Let $u,v\in V(C)$.
We first bound the expansion of $f'$.
Since $K\subset K'$, it follows what $d_{G\setminus K}(u,v)\leq d_{G\setminus K'}(u,v)$, and thus
\begin{align}
|f'(u)-f'(v)| &= (4c^2+1) \cdot |f(u)-f(v)| \notag \\
 &\leq (4c^3+c) \cdot d_{G\setminus K}(u,v) \notag \\
 &\le (4c^3+c) \cdot d_{G\setminus K'}(u,v)
 \label{eq:f_prime_expansion}.
\end{align}

It remains to show that $f'$ is non-contractive.
Let $P$ be the shortest path between $u$ and $v$ in $G\setminus K$.
Let us first assume that $u,v\in V_{\safe}$; we will consider the general case later.
If $z_L$ is defined and $P\cap V_{\outerr}(z_L)$, we first construct a new path $P'$ that avoids $V_{\outerr}(z_L)$, as follows.
When traversing $P$ starting from $u$, let $u_1$ be the last vertex before visiting $V_{\outerr}(z_L)$ for the first time;
let also $u_2$ be the first vertex visited immediately after leaving $V_{\outerr}(z_{L})$ for the last time.

Since the expansion of $f$ is at most $c$, it follows that
\begin{align*}
f(u_1) & \in  (f(z_L)+2c^2, f(z_L)+2c^2+c], \\
f(u_2) & \in  (f(z_L)+2c^2, f(z_L)+2c^2+c],
\end{align*}
and thus
\begin{align}
d_{G\setminus K}(u_1,u_2) &\leq |f(u_1)-f(u_2)| \leq c. \label{eq:tasos10}
\end{align}
Let $W$ be the shortest path between $u_1$ and $u_2$ in $G\setminus K$.
Since every edge of $W$ is stretched by at most a factor of $c$ in $f$, it follows by \eqref{eq:tasos10} that $W$ cannot enter $V_{\inner}(z_L)$, and thus $W \subseteq G\setminus K'$.
Therefore
\begin{align*}
d_{G\setminus K'}(u_1,u_2) &= d_{G\setminus K}(u_1,u_2) \leq c
\end{align*}
We can replace $P$ be the path 
\[
P':=P[u,u_1] \circ W \circ P[u_2,v],
\]
which does not intersect $v_{\outerr}(z_L)$.
We obtain that
\begin{align*}
\length(P') &= \length(P[u,u_1]) + \length(W) + \length(P[u_2,v])\\
 &\leq c +  \length(P) \\
 &\leq c + d_{G\setminus K}(u,v).
\end{align*}

Next, if $z_R$ exists and $P'\cap V_{\outerr}(z_R)\neq \emptyset$, then via a symmetric process we can replace $P'$ by a new path $P''$ between $u$ and $v$ in $G\setminus K$ avoids $V_{\outerr}(z_R)\cup V_{\outerr}(z_L)$, with 
\begin{align*}
\length(P'') &\leq \length(P') + c \leq \length(P) + 2c.
\end{align*}
This implies that $P'\subseteq G\setminus K'$.

We therefore obtain
\begin{align}\label{eq:f_prime_contraction}
|f'(u)-f'(v)| &= (4c^2+1) \cdot |f(u)-f(v)| \notag \\
 &\geq (4c^2+1) \cdot  d_{G\setminus K} (u, v) \notag \\
 &\geq (4c^2+1) \cdot (d_{G\setminus K'}(u,v) - 2c) \notag \\
 &> d_{G\setminus K'}(u,v).
\end{align}
By \eqref{eq:f_prime_expansion} and \eqref{eq:f_prime_contraction}
we obtain that $f$ is a $(4c^3+c)$-embedding of $G\setminus K'$, as required.

It remains to consider the case where either $u\in V_{\exposed}$, or $v\in V_{\exposed}$.
Let $Q$ be a shortest path between $u$ and $v$ in $G\setminus K'$.
If $Q\cap V_{\safe}=\emptyset$, then $\length(Q) \leq 4c^2-2c$, and thus
$d_{G\setminus K'}(u,v) \leq (4c^2-2c) \leq   (4c^2-2c) \cdot d_{G\setminus K'}(u,v)$, which implies that $f'$ is non-contractive, as required.
We may therefore assume for the remainder of the proof that $Q\cap V_{\safe}\neq \emptyset$.
If $u\in V_{\exposed}$, then we may assume w.l.o.g.~that $u\in V_{\outerr}(z_L)$.
When traversing $P$ starting from $u$, let $u_1$ be the first vertex visited immediately after leaving $V_{\outerr}(z_L)$.
When traversing $Q$ starting from $u$, let $u_2$ be the first vertex visited in $V_{\safe}$.
By an argument identical to the one used in the previous case, 
we can obtain a new path between $u$ and $v$, given by $Q[u,u_2]\circ W \circ P[u_1,v]$,
where $\length(Q[u,u_2])\leq 2c^2-c$ 
(since all vertices in $Q[u,u_2]$ 
except the last one are contained in the rightmost segment of 
 $V_{\outerr}(z_L)$),
 $W\subseteq G\setminus K'$, and $\length(W)\leq c$ (as in the previous case).
We thus obtain a path of length at most $d_{G\setminus K}(u,v) + 2c^2$.
If $v\in V_{\exposed}$, we repeat the above process after exchanging $u$ and $v$.
We thus arrive at a path between $u$ and $v$ of length at most $d_{G\setminus K}(u,v) + 4c^2\leq (4c^2+1) \cdot d_{G\setminus K}$, which does not intersect $V_{\inner}$, and thus it is contained in $G\setminus K'$.
It follows that $f'$ is non-contractive, and thus a $(4c^3+c)$-embedding, which concludes the proof.
\end{proof}

\section{Density Reduction}\label{sect:Density_Reduction}

This Section is devoted to the algorithm used for reducing the local density of the input graph, which is summarized in Lemma \ref{lem:density_reduction}.
We first present the algorithm and then its analysis.

\subsection{The algorithm for density reduction}
Let us describe the algorithm for reducing the density of a graph.
The algorithm takes as input a graph $G$ and some $k\ge 0$, $c\ge 1$, and outputs some $Y\subseteq V(G)$, such that $\Delta(G\setminus Y)\le c$.
This is summarized in Algorithm \ref{alg:density}.
 
\begin{algorithm}[H]
\caption{SPARSIFY}\label{Density reduction}
\begin{algorithmic}[1]
\Procedure{SPARSIFY}{$G$, $c$}
\If{ $\Delta(G)\leq c$ }
 \Return $\emptyset$
\Else
 \State Let $X$ be a $3/4$-balanced vertex separator of $G$ computed by Theorem \ref{thm:v_separator_approx}.
 \State Let $G_1,\ldots,G_t$ be the connected components of $G\setminus X$.
 \Return $X \cup \left( \bigcup_{i=1}^t \text{SPARSIFY}(G_i,c) \right)$
\EndIf
\EndProcedure
\end{algorithmic}
\label{alg:density}
\end{algorithm}

\subsection{Analysis of the algorithm for density reduction}

We now analyze the algorithm described above.
We first recall the following result from \cite{badoiu2005approximation}.

\begin{lemma}[B\u{a}doiu \etal~\cite{badoiu2005approximation}]\label{Density reduction claim}
If $G$ admits a $c$-embedding into the line then $\Delta(G)\leq c$.
\end{lemma}

The following establishes the existence of small balanced separators.

\begin{lemma}\label{lem:v_separator}
Let $G$ be a graph such that $G$ admits a $(k,c)$-embedding into the line. Let $Z\subseteq V(G)$ with $|Z|=k$ be such that $G\setminus Z$ is $c$-embeddable into the line. Then any $H\subseteq G$ contains a $2/3$-balanced vertex separator of size at most $c+|Z\cap V(H)|$.
\end{lemma}
\begin{proof}
Let $f:G\setminus Z\to \R$ be an embedding with distortion $c.$ Let $V(H)=\{v_1,\ldots,v_h\}$, and assume w.l.o.g.~that
\[
f(v_1)<f(v_2)<\ldots<f(v_h)
\]
Let $X=\{v_{\lfloor h/3 \rfloor+1}, v_{\max\{h, \lfloor h/3\rfloor +c+1\}}\}$.
By lemma \ref{C-cut} we get that $X$ is a balanced separator of $H\setminus Z=H\setminus (Z\cap V(H))$.
Therefore, $W:=X\cup (V(H) \cap Y)$ is a balanced separator for $H$, with
$|W| = |X|+|Z\cap V(H)| \leq c+|Z\cap V(H)|$, as required.
\end{proof}

We are now ready to prove the main result of this Section.

\begin{proof}[Proof of Lemma~\ref{lem:density_reduction}.]
It is immediate that the output, $Y$, of the procedure SPARSIFY is such that $\Delta(G\setminus Y)\leq c$. Also, if $\Delta(G)\le c$, the algorithm outputs $Y=\emptyset$.

It thus remains to bound $|Y|$.
Fix some $K\subseteq V(G)$, with $|K|=k$, and some $c$-embedding $f$ of $G\setminus K$ into the line.
Consider some recursive call of procedure SPARSIFY$(H,c)$, for some $H\subseteq G$.
If $H\cap K=\emptyset$, then $H\subseteq G\setminus K$, and thus $\Delta(H)\leq \Delta(G\setminus K)\leq c$, where the last inequality follows by lemma \ref{Density reduction claim}.
Therefore, procedure SPARSIFY computes a balanced separator, $X_H$, only if $H$ intersects $K$.
By lemma \ref{lem:v_separator} and Theorem \ref{thm:v_separator_approx} it follows that 
\begin{align*}
|X_H| &\leq O\left( \sqrt{\log{n}} \cdot (c+|K\cap V(H)|)\right)\\
 &\leq O\left( |K\cap V(H)|\cdot c \cdot \sqrt{\log{n}}\right).
\end{align*}
We charge the vertices in $X_H$ to the vertices in $K\cap H$;  thus every vertex in $K\cap H$ receives at most $O\left( \sqrt{\log{n}}\right)$ units of charge.
Since any two subgraphs on the same level of the recursion are disjoint, it follows that each vertex in $K$ receives at most $O\left( c \sqrt{\log{n}}\right)$ units of charge per level of the recursion.
Since each separator is $3/4$-balanced, it follows that the depth of the recursion is at most $\log_{4/3} n$. Thus, every vertex in $K$ receives at most $\log_{4/3}{n}\cdot O(c\log{n})=\beta\cdot c \log_{4/3}^{3/2}{n}$ units of charge throughout the execution of the procedure SPARSIFY. The constant $\beta$ comes from the bound on the size of the vertex separator computed by Theorem~\ref{thm:v_separator_approx}. Hence, if $Y>\beta \cdot k c\log_{4/3}^{3/2} n$, then we have certified that $G$ does not admit a $(k,c)$-embedding into the line, which concludes the proof.
\end{proof}

\section{Eliminating large metrical cycles}\label{sect:Embed_into_Forest}

In this Section we describe and analyze the algorithm for eliminating large metrical cycles.

\subsection{The algorithm}

%We begin by describing the algorithm.
The input consists of a graph $G$, some $c\geq 1$, and $k\geq 0$.
%The algorithm will either correctly decide that $G$ does not admit a $(k,c)$-embedding into the line, or it will compute some $K'\subseteq V(G)$, some forest $F$, and an embedding $f$ of $G\setminus K'$ into $F$.
The algorithm proceeds in steps, that are formally described below.

\textbf{Algorithm for eliminating large metrical cycles:}
\begin{description}
\item{\textbf{Step 1.}}
Compute a $c$-net $N$ of $G$.
\item{\textbf{Step 2.}}
Compute a Voronoi partition ${\cal P}$ of $G$ centered at $N$, and the corresponding $c$-minor $H$ of $G$.
\item{\textbf{Step 3.}}
Using the algorithm from Theorem \ref{thm:MFVS} compute a 2-approximate solution $S$ to the Minimum Feedback Vertex Set problem on $H$.
If $|S|>2k$, then decide that $G$ does not admit a $(k,c)$-embedding into the line.
\end{description}
This completes the description of the algorithm.

\subsection{Analysis}

First, we prove the following statement about embeddability into a subgraph of a $c$-minor. 

\begin{lemma}\label{lem:minor_embed}
Let $G$ be a graph, $R>0$, 
let $N$ be a $R$-net in $G$, let ${\cal P}$ be a $R$-partition centered at $N$, and let $H$ be the $R$-minor of $G$ induced by ${\cal P}$.
Let $X\subset N$, and let
\[
Y = \bigcup_{x\in X} {\cal P}(x).
\]
Then the metric space $(N\setminus X, d_{G\setminus Y})$ admits a $(2R+1)$-embedding into $H\setminus X$.
Moreover, this embedding can be computed in polynomial time.
\end{lemma}

\begin{proof}
Let $u,v\in N\setminus X$.
Let $Q$ be a $u$-$v$ shortest path in $G\setminus Y$.
When traversing $Q$ starting from $u$ let $C_1,\ldots,C_{\ell}$ be the sequence of clusters of ${\cal P}$ visited.
For each $i\in [\ell]$ let $q_i$ be the center of $C_i$; that is, $C_i = {\cal P}(q_i)$.
Since for all $i\in [\ell-1]$ there is an edge in $G\setminus Y$ between some vertex in $C_i$ and some vertex in $C_{i+1}$, it follows that there also exists an edge in $H\setminus X$ between $q_i$ and $q_{i+1}$.
Therefore $Q'=q_1,\ldots,q_{\ell}$ is a path in $H\setminus X$.
We thus obtain
\begin{align}
d_{H\setminus X}(u,v) &\leq \length(Q') \leq \length(Q)  = d_{G\setminus Y}(u,v). \label{eq:minor_embed_1}
\end{align}

Let $W=w_1,\ldots,w_t$ be a $u$-$v$ shortest path in $H\setminus X$.
Since each cluster in ${\cal P}$ has radius at most $R$, it follows that for all $i\in [t-1]$ there exists a $w_i$-$w_{i+1}$ path in $G\setminus Y$ of length at most $2R+1$.
Concatenating all these paths we obtain a $u$-$v$ path $W'$ in $G\setminus Y$ of length at most $(t-1)\cdot (2R+1)$.
Thus
\begin{align}
d_{G\setminus Y}(u,v) &\leq \length(W') \leq (2R+1)(t-1) =  (2R+1) d_{H\setminus X}(u,v).
 \label{eq:minor_embed_2}
\end{align}
Combining \eqref{eq:minor_embed_1} and \eqref{eq:minor_embed_2} the assertion follows.
\end{proof}

%\subsection{Analysis}

We recall the Borsuk-Ulam Theorem \cite{borsuk1933drei}.

\begin{theorem}[Borsuk-Ulam Theorem \cite{borsuk1933drei}]\label{thm:BU}
Let $d\geq 1$, and let $\mathbb{S}^d$ denote the $d$-dimensional sphere.
Let $f:\mathbb{S}^d\to \mathbb{R}^d$ be a continuous map.
Then there exists $x\in \mathbb{S}^d$, such that $f(x)=f(-x)$.
\end{theorem}

The following is a simple consequence of Theorem \ref{thm:BU}.
A similar argument is used in  \cite{badoiu2005approximation}.

\begin{lemma}\label{lem:cycle-non-embed}
Let $C$ be a cycle and let $f:V(C)\to \mathbb{R}$ be an injective map.
Then there exist  $u,v,w\in V(C)$, such that $\{u,v\}\in E(C)$, and 
$f(u)<f(w)<f(v)$.
\end{lemma}

\begin{proof}
Suppose that $C$ is the $n$-cycle for some $n\in \mathbb{N}$.
We identify the vertices in $C$ with distinct points in $\mathbb{S}^1$, so that the points appear in the same order as in $C$ along a clockwise traversal of $\mathbb{S}^1$.
For each $\{x,y\}\in E(C)$ ther exists an arc $A_{x,y}$ in $\mathbb{S}^1$ that does not contain any other vertex in $C$; we extend $f$ to $A_{x,z}$ affinely.
After repeating for all edges in $C$, we obtain a continuous map $f:\mathbb{S}^1\to \mathbb{R}^1$.
By Theorem \ref{thm:BU} we get that there exists $x\in \mathbb{S}^1$ with $f(x)=f(-x)$.
This means that there exist two edges in $C$ whose images in $f$ span overlapping intervals in $\mathbb{R}^1$.
Since $f$ is injective on $V(C)$ this implies that one endpoint is contained inside the interval of the other edge, which concludes the proof.
\end{proof}

We next establish the existence of a small feedback vertex set in the minor computed by the algorithm.

\begin{lemma}\label{lem:FVS_exists}
Let $G$ be a graph, $c\geq 1$, $k\geq 0$, such that $G$ admits a $(k,c)$-embedding into the line.
Let $H$ be a $R$-minor of $G$, for some $R\geq c$.
Then there exists a feedback vertex set $X$ in $H$ with $|X|\leq k$.
\end{lemma}

\begin{proof}
Let ${\cal P}$ be the $R$-partition of $G$ such that $H$ is the $R$-minor of $G$ induced by ${\cal P}$.
Since $G$ admits a $(k,c)$-embedding into the line, it follows that there exists some $Y\subseteq V(G)$, with $|Y|\leq k$, such that $G\setminus Y$ admits a $c$-embedding $f$ into the line.

Let $X$ be the set of all $v\in V(H)$, such that $Y$ intersects the cluster in ${\cal P}$ centered at $v$; that is
\[
X = \{v\in V(H) : {\cal P}(v)\cap Y\neq \emptyset\}.
\]
Since ${\cal P}$ is a partition, it is immediate that $|X|\leq |Y|\leq k$.
It therefore remains to show that $H\setminus X$ is acyclic.
Suppose, for the sake of contradiction, that $H\setminus X$ is not acyclic.
Let $C$ be a cycle in $H\setminus X$.
By Lemma \ref{lem:cycle-non-embed} there exist $u,v,w\in V(C)$, such that $\{u,v\}\in E(C)$, and \[
f(u)<f(w)<f(v).
\]

Since $\{u,v\}\in E(C)$, and $C\subseteq H$, it follows that $\{u,v\}\in E(H)$.
Since $H$ is $R$-minor, it follows that there exists a path $Q$ between $u$ and $v$, with $Q\subseteq {\cal P}(u)\cup {\cal P}(v)$.
When traversing $Q$ starting from $u$ let $u'$ be the last vertex visited with $f(u')<f(w)$;
let also $v'$ be the vertex visited immediately after $u'$.
We have
\[
f(u')<f(w)<f(v').
\]

Since $H$ is a $R$-minor and $u'\notin{\cal P}(w)$,
it follows that
$d_G(w,u') \geq d_G(u,u')$.
By the definition of a $R$-partition we have that $d_G(u,w)>c$, and therefore
\[
d_G(u',w)>R/2.
\]
Similarly, we obtain
\[
d_G(v',w)>R/2.
\]
Since $f$ is non-contracting, we obtain
\begin{align*}
|f(u')-f(v')| &= |f(u')-f(w)| + |f(w)-f(v')|\\
 &\geq d(u,w) + d(w,v)\\
 &> R/2+R/2 = R \geq c,
\end{align*}
which contradicts the fact that $f$ has expansion at most $c$, and concludes the proof.
\end{proof}

We are now ready to prove the main result of this Section.

\begin{proof}[Proof of Lemma~\ref{lem:forest}]
By Lemma \ref{lem:FVS_exists}, either $G$ does not admit a $(k,c)$-embedding into the line, or there exists $X\subseteq V(H)$, with $|X|\leq k$, such that $H\setminus X$ is acyclic.
Using the algorithm from Theorem \ref{thm:MFVS} we compute in Step 3 a 2-approximation $S\subseteq V(H)$ to the Minimum Feedback Vertex Set in $H$. 
Therefore, if $|S|>2k$, then we can terminate with outcome (1), and otherwise terminate with outcome (2), which completes the proof.
\end{proof}

\section{Eliminating large metrical tripods}\label{sect:Tripod_Elimin}

In this Section we present and analyze the procedure for eliminating large metrical tripods.
We begin by showing that large tripods are an obstruction to embeddability into the line.
This is summarized in Lemma \ref{lem:tripod_distortion}.

\begin{proof}[Proof of Lemma~\ref{lem:tripod_distortion}]
Let $f$ be a non-contractive embedding of $J$ into the line. Let $v$ be the common endpoint of $P_1$, $P_2$, $P_3$.
For each $i\in [3]$ let $v_i$ be the other endpoints of $P_i$.
We may assume w.l.o.g.~(by change of indices) that 
\[
f(v_1) < f(v_2) < f(v_3).
\]
Let $Q$ be the unique $v_1$-$v_3$ path in $J$.
It follows that there exists $\{u,w\}\in E(Q)$, such that
\[
f(u) < f(v_2) < f(w).
\]
This implies that
\begin{align*}
|f(u)-f(w)| &= |f(u)-f(v_2)| + |f(v_2)-f(w)| \\
 &\geq d_G(u,v_2) + d_G(v_2, w) \\
 &\geq 2R \\
 &= 2R d_J(u,w).
\end{align*}
Therefore the distortion of $f$ is at least $2R$, 
which concludes the proof.
\end{proof}

The above easily implies the following results, which asserts the existence of a small set of vertices whose removal eliminates all large tripods.

\begin{lemma}\label{lem:tripod_elimination_existence}
Let $F$ be a forest that admits a $(k,c)$-embedding into the line.
Then there exists some $X\subseteq V(F)$, with $|X|\leq k$, such that $F\setminus X$ does not contain any $(c/2+1)$-tripod as a subgraph.
\end{lemma}

\begin{proof}
Since $F$ admits a $(k,c)$-embedding into the line, it follows that there exists some $X\subseteq V(F)$, with $|X|\leq k$, such that $F\setminus X$ admits a $c$-embedding into the line.
It suffices to show that $F\setminus X$ does not contain any $(c/2+1)$-tripods.
Suppose, for the sake of contradiction, that $F\setminus X$ contains some $(c/2+1)$-tripod $J$.
Since $(V(J), d_J)$ is a submetric of $(V(F)\setminus X, d_{F\setminus X})$, it follows that $J$ admits a $c$-embedding into the line, which contradicts Lemma \ref{lem:tripod_distortion}, and concludes the proof.
\end{proof}

Now are now ready to prove the main result of this Section.

\begin{proof}[Proof of Lemma \ref{lem:tripod_eliminaton}]
Any tripod $T\subseteq F$ can be uniquely specified by selecting its root and its three leaves.
Therefore, there are at most $O(|V(F)|^4)$ distinct tripods in $F$.
Moreover, the set of all tripods, ${\cal T}$, can be enumerated in polynomial time.
We form an instance of the Minimum Set Cover problem with universe $U={\cal T}$.
We also let
\[
{\cal C} = \bigcup_{v\in V(F)} \{C_v\},
\]
where
\[
C_v = \{T \in {\cal T} : v\in V(T)\}.
\]
It is immediate that for any $Y\subseteq V(F)$, $F\setminus Y$ contains no $R$-tripods iff $\bigcup_{v\in Y} C_v = U$.
Therefore, computing a minimum-cardinality subset of vertices of $F$ whose deletion removes all $R$-tripods, is equivalent to solving the Minimum Set Cover instance on $(U, {\cal C})$.
The result now follows from Theorem \ref{thm:SetCoverApprox}.
\end{proof}

\section{Embedding Trees Without Large Tripods into the Line}\label{sect:Trees_no_large_Tripods}

This Section is devoted to proving Lemma \ref{lem:embed_trees_no_tripods}, which asserts that any tree with no large tripods admits a low-distortion embedding into the line.

\begin{proof}[Proof of Lemma~\ref{lem:embed_trees_no_tripods}]
Since $T$ is a tree, we can compute in polynomial time a longest path $Q$ in $T$.
Let $Q=v_1,\ldots,v_t$.
Let ${\cal P}$ be a Voronoi partition centered at $V(Q)$.
Since $T$ does not contain any $R$-tripod as a subgraph, it follows that for all $u\in V(T)$, there exists some $v\in V(Q)$, with $d_{T}(u,v) < R$.
Therefore, for each $v_i\in V(Q)$, we have
\begin{align*}
|{\cal P}(v_i)| &\leq |\ball_{T}(v_i, R-1)| \\
 &\leq \Delta(T) \cdot 2(R-1) + 1 \\
 &\leq \Delta(T) \cdot 2R - 1
\end{align*}

By the definition of a graphical Voronoi partition we have that for all $i\in [t]$, the vertex-induced subgraph $T_i:=T[{\cal P}(v_i)]$ is connected, and thus $T_i$ is a subtree of $T$.
Let $W_i$ be a closed walk in $T_i$ that visits all vertices in $T_i$, obtained by duplicating every edge (or, equivalently, the walk obtained by any traversal of $T_i$).
Since every edge in $T_i$ is traversed twice, we have $\length(T_i)=2(|V(T_i)|-1)$.
Let $W_i=w_{i,1},\ldots,w_{i,t_i}$.

We define the embedding $f_i:V(T_i)\to \mathbb{R}$ as follows.
For each $v\in V(T_i)$, we define
\[
f_i(v) = \min\{j\in [t_i] : v=w_{i,j}\}.
\]

We combine the mappings $f_1,\ldots,f_t$ into a mapping $f:V(G)\to \mathbb{R}$.
Informally, this is done by translating each $f_i$ so that for all $i\in [t-1]$, the image of $f_i$ appears to the left of the image of $f_{i+1}$, and there is a gap of length $2R$ between these two images.

Formally, for each $u\in {\cal P}_{v_i}$, we set $f(u)=L_i + f_i(u)$, 
where
\[
L_i = \left\{\begin{array}{ll}
0 & \text{ if } i=0\\
L_{i-1} + \max_{z\in {\cal P}(v_{i-1})} \{ f_{i-1}(z) \} + 2R & \text{ otherwise}\\
\end{array}\right.
\]
This completed the definition of the embedding $f$.

It remains to bound the distortion of $f$.
For vertices that lie in the same cluster in ${\cal P}$, the map is non-contractive since the distance in the embedding is at least the distance in some walk $W_i$, which is at least the distance in $T$.
Moreover, the expansion is upper bounded by the length of the walk, which is at most $\Delta(T)\cdot (2R-1)$.

Next, let us consider $p,q\in V(T)$ that fall in different clusters in ${\cal P}$.
Suppose that $p\in {\cal P}(v_i)$, and $q\in {\cal P}(v_j)$, for some $i,j\in [t]$, with $i<j$.
We have
\begin{align*}
|f(p)-f(q)| &\leq 2R (j-1) + \sum_{r=i}^{j} \length(W_i) \\
 &\leq (j-1)2R + (j-i+1) \Delta(T) \cdot (2R-1)\\
 &\leq (j-i) \cdot O(\Delta(T) \cdot R) \\
 &= d_T(v_i,v_j) \cdot O(\Delta(T) \cdot R) \\
 &\leq d_T(p,q) \cdot O(\Delta(T) \cdot R).
\end{align*}
Moreover
\begin{align*}
|f(p)-f(q)| &\geq 2R(j-i) + 1\\
 &\geq 2R + (j-i)\\
 &\geq d_T(p,v_i)+d_T(v_i,v_j)+d_T(v_j,q)\\
 &= d_T(p,q).
\end{align*}

We conclude that, in all cases, $f$ is non-contractive and has expansion at most $O(\Delta(T)\cdot R)$, as required.
\end{proof}

\section{Properties of $R$-minors and $R$-partitions}\label{sec:minors}

In this Section we prove Lemmas \ref{lem:properties_of_voron} and \ref{lem:distance_in_graph_vs_minor}, which establish some basic properties of $R$-minors and $R$-partitions.

\begin{proof}[Proof of Lemma~\ref{lem:properties_of_voron}]
We first show (1).
    Since by deleting vertices the shortest-path distances cannot increase, we have that for all $u,v\in N'$, $d_{G'}(u,v)\geq d_G(u,v) > R$.
    It thus remains to show that for any $x\in V(G')$ there exists $v\in N'$ such that     $d_{G'}(x,v)\le R$.
    Consider an arbitrary $x\in V(G')$. Let $v\in N$ be such that $x\in {\cal P}(v).$ 
    Since the shortest path between $v$ and $x$ in $G$ is contained in ${\cal P}(v)$, it follows that 
    \begin{align*}
         d_{G'}(x,v)&\le d_{G'[{\cal P}(v)]}(x,v) = d_{G[{\cal P}(v)]}(x,v) =d_{G}(x,v)\le c,
    \end{align*}
    which implies that $N'$ is a $c$-net of $G'$.

Next, we show (2).
Since for all $v\in N'$, we have ${\cal P}'(v)={\cal P}(v)$, it follows that ${\cal P}'$ is a partition of $V(G')$.
Since by (1) $N'$ is a $R$-net of $G'$, and for all $v\in N'$, and for all $x\in {\cal P}'(v)$ we have $d_{G'}(v,x)\leq R$, it follows that ${\cal P}'$ is a $R$-partition of $G'$ centered at $N'$.

Finally, we show (3). Let $\widetilde{H}$ be the $R$-minor of $G'$ induced by $\mathcal{P}'$.
 We prove that $V(H')=V(\widetilde{H})$ and $E(H')=E(\widetilde{H})$. For the first equality, observe that 
 \begin{align*} V(H')&=V(H\setminus Y)
 =N\setminus Y
 = N'
 =V(\widetilde{H}).
 \end{align*}
 It remains to show that $E(H')=E(\widetilde{H})$.
    Consider an arbitrary $\{u,v\}\in E(H')$. Since $H'
\subseteq H$ we have that $\{u,v\}\in E(H)$. Then there must exist a path $P\subseteq G$ between $u$, $v$ with $P\subseteq {\cal P}(u)\cup {\cal P}(v)$. Since $u,v\in V(H\setminus Y)=N'$ we have that $\mathcal{P}(u)=\mathcal{P}'(u)$, $\mathcal{P}(v)=\mathcal{P'}(v)$. Thus, $P\subseteq \mathcal{P}'(u)\cup\mathcal{P}'(v)$ which yields $\{u,v\}\in E(\widetilde{H})$. Now consider an arbitrary $\{u,v\}\in E(\widetilde{H})$; it induces a path $Q\subseteq G'$ between $u$, $v$ such that $Q\subseteq \mathcal{P}'(u)\cup\mathcal{P}'(v)$. Since $\mathcal{P}'(u)=\mathcal{P}(u)$, $\mathcal{P}'(v)=\mathcal{P}(v)$ we obtain $\{u,v\}\in E(H)$. Then from $u,v\in N'=N\setminus Y$ we have $\{u,v\}\in E(H\setminus Y)=E(H')$  which concludes the proof.
\end{proof}

\begin{proof}[Proof of Lemma~\ref{lem:distance_in_graph_vs_minor}]
Let $P\subseteq G$ be a shortest path between $u,v$ and let 
\[
J:=\{w\in N: P\cap {\cal P}(w)\neq \emptyset\}.
\]
Let $Q\subseteq H$ be a shortest path between $u,v$. We claim that
\[
\length(Q)\le |J|-1\le \length(P). 
\]
Assume for contradiction that $\length(Q)>|J|-1$. Consider arbitrary $\{x_1,x_2\} \in E(P)$ such that $x_1\in {\cal P}(w_1), x_2\in {\cal P}(w_2)$ for $w_1\neq w_2$; hence $\{w_1,w_2\}\in E(H)$. Therefore, $P$ induces a walk $W\subseteq H$ such that $v,u\in V(W)$. Hence, there is a path $Q'\subseteq W$ such that $v,u\in V(Q')$; note that $\length(Q)\le |V(W)|-1=|J|-1$. Thus,
\begin{align*}
    \length(Q')&\le |J|-1 < \length(Q) =d_{H}(v,u).
\end{align*}
which gives a contradiction, and concludes the  proof.
\end{proof}

%\bibliographystyle{plain}
%\bibliography{bibfile}
\bibliographystyle{plain}
\bibliography{bibfile}

\end{document}